\documentclass[a4paper]{article}
\usepackage{latexsym,amsthm,amsmath,amssymb}
\usepackage{subcaption}

\theoremstyle{plain}
\newtheorem{theorem}{Theorem}
\newtheorem{claim}[theorem]{Claim}
\newtheorem*{theorem*}{Theorem}
\newtheorem{lemma}[theorem]{Lemma}
\newtheorem*{lemma*}{Lemma}
\newtheorem{corollary}[theorem]{Corollary}
\newtheorem{observation}[theorem]{Observation}
\newtheorem*{observation*}{Observation}

\newtheorem*{conjecture*}{Conjecture}
\newtheorem{reduction}{Reduction Rule}
\theoremstyle{definition}

\usepackage{graphicx}
\usepackage{tikz}
\usepackage{tkz-berge}
\usepackage{comment}
\usepackage{authblk}

\usepackage{complexity}
\newclass{\Hard}{hard}
\newclass{\Hness}{hardness}
\newcommand{\NPH}{\NP\text{-}\Hard}
\newcommand{\NPHness}{\NP\text{-}\Hness}

\newclass{\para}{para}
\newclass{\Complete}{complete}
\newclass{\Cness}{completeness}

\newfunc{\dist}{dist}
\newfunc{\girth}{girth}
\newfunc{\nd}{nd}
\newfunc{\YES}{YES}
\newfunc{\NOi}{NO}
\newfunc{\ff}{ff}
\newfunc{\dc}{dc}
\newfunc{\dcc}{d\overline{c}}
\newfunc{\ml}{ml}
\newfunc{\cf}{cf}
\newfunc{\tw}{tw}
\newfunc{\ext}{ext}
\newfunc{\pat}{pat}
\newfunc{\parts}{parts}
\newfunc{\rev}{rev}
\newfunc{\ttwins}{t-twins}
\newfunc{\ftwins}{f-twins}

\newcommand{\probl}[3]{
  \begin{flushleft}
    \fbox{
      \begin{minipage}{0.95\linewidth}
        \noindent {\pname{#1}}\\
        {\bf Instance:} #2\\
        {\bf Question:} #3
      \end{minipage}}
    \medskip
  \end{flushleft}
}

\newcommand{\pname}[1]{\textsc{#1}}
\newcommand{\bigO}[1]{{\mathcal{O}\!\left(#1\right)}}

\newcommand{\floor}[1]{\left\lfloor#1\right\rfloor}
\newcommand{\inners}{1.2pt}
\newcommand{\outers}{1pt}

\newcommand{\tp}{{t(P)}}

\newenvironment{cproof}[1]{\par \noindent \textit{Proof of Claim~#1}.}{\hfill $\blacksquare$\newline}

\title{Parameterized algorithms for locating-dominating sets}
\date{}

\author[1]{M\'arcia R. Cappelle}

\author[2]{Guilherme C. M. Gomes}
\author[2]{Vinicius F. dos Santos}
\affil[1]{Instituto de Inform\'atica, Universidade Federal de Goi\'as -- Goi\^ania, Brazil}
\affil[2]{Departamento de Ci\^encia da Computa\c{c}\~{a}o, Universidade Federal de Minas Gerais -- Belo Horizonte, Brazil}

\begin{document}

\maketitle

\begin{abstract}
    A locating-dominating set $D$ of a graph $G$ is a dominating set of $G$ where each vertex not in $D$ has a unique neighborhood in $D$, and the \pname{Locating-Dominating Set} problem asks if $G$ contains such a dominating set of bounded size.
    This problem is known to be \NPH\ even on restricted graph classes, such as interval graphs, split graphs, and planar bipartite subcubic graphs.
    On the other hand, it is known to be solvable in polynomial time for some graph classes, such as trees and, more generally, graphs of bounded cliquewidth.
    While these results have numerous implications on the parameterized complexity of the problem, little is known in terms of kernelization under structural parameterizations.
    In this work, we begin filling this gap in the literature.
    Our first result shows that \pname{Locating-Dominating Set}, when parameterized by the solution size $d$, admits no $2^{o(d \log d)}$ time algorithm unless the Exponential Time Hypothesis fails; as a corollary, we also show that no $n^{o(d)}$ time algorithm exists under ETH, implying that the naive \XP\ algorithm is essentially optimal.
    We present an exponential kernel for the distance to cluster parameterization and show that, unless $\NP \subseteq \coNP/\poly$, no polynomial kernel exists for \pname{Locating-Dominating Set} when parameterized by vertex cover nor when parameterized by distance to clique.
    We then turn our attention to parameters not bounded by neither of the previous two, and exhibit a linear kernel when parameterizing by the max leaf number; in this context, we leave the parameterization by feedback edge set as the primary open problem in our study.
\end{abstract}

%------------------------------------------------------------------

\section{Introduction}
\label{sec:int}
%------------------------------------------------------------------

A dominating set in a graph $G$ is a set $D$ of vertices of $G$ such that every vertex outside $D$ is adjacent to a vertex in $D$. A locating–dominating  set is a dominating set $D$
that locates/distinguishes all the vertices in the sense that every vertex not in $D$ has a unique neighborhood in $D$.
Formally, we define the problem as follows:

\probl{Locating-Dominating Set}{A graph $G$ and an integer $d$.}{Does $G$ admit a locating-dominating set of size at most $d$?}

This problem was introduced and first studied by Slater~\cite{slater1987domination}, and since then has been extensively studied in the literature~\cite{foucaud2015,foucaud2017-i,hernando2019,junnila2019,rall1984}.
It was proven to be NP-complete for
bipartite graphs by Charon et al.~\cite{charon2003} and 
interval graphs and permutation graphs by Foucaud et al.~\cite{foucaud2017-i}. 
These results were later extended for 
interval graphs of diameter 2 and permutation graphs of diameter 2 in~\cite{foucaud2017-ii}, and to subcubic planar bipartite graphs in~\cite{foucaud2015}.
On the other hand, it is linear-time solvable for trees~\cite{auger2010minimal,slater1987domination} and, more generally, for graphs of bounded clique-width; the latter result is obtained through a direct application of Courcelle’s theorem \cite{courcelle1990}.
In terms of approximation, \pname{Locating–Dominating Set} can be approximated to a logarithmic factor, but no sub-logarithmic approximation algorithm exists under standard complexity hypotheses.
These inapproximability results were extended to bipartite graphs, split graphs and co-bipartite graphs \cite{foucaud2015}.
We refer to \cite{lobstein2012watching} for an on-line bibliography on this topic and related notions.
In terms of parameterized complexity, these hardness results imply that the problem is \para\NPH\ when parameterized by diameter, distance to interval, distance to split, maximum degree, genus, mim-width~\cite{mimwidth_interval}, twin-width~\cite{twinwidth1}, and minimum clique cover; moreover, parameterizing by the natural parameter $d$ (the size of the solution) would be trivially \FPT\ since $d \in \Omega(\log n)$.
On the other hand, the linear-time algorithm for graphs of bounded cliquewidth implies that the problem is fixed-parameter tractable under several parameterizations such as cliquewidth, distance to cograph, treewidth, feedback vertex set, max leaf number, and vertex cover.
However, to the best of our knowledge, no kernelization results for structural parameterizations of \pname{Locating-Dominating Set} have been presented so far.

\medskip

\noindent {\it Our contributions.}
In this work we seek to partially fill the space currently present in the parameterized complexity literature.
We begin by showing that, unless the Exponential Time Hypothesis~\cite{eth} fails, \pname{Locating-Dominating Set} does not admit algorithms that run in $2^{o(d \log d)}$ nor in $n^{o(d)}$, which implies that the naive \XP\ algorithm for the natural parameter is essentially optimal.
We then proceed in full to kernelization, by first showing an exponential kernel for the distance to cluster parameter, and then proving that, unless $\NP \subseteq \coNP/\poly$, no polynomial kernel exists when jointly parameterizing by vertex cover and maximum size of the solution, nor when parameterizing by distance to clique and maximum size of the solution.
On the positive side, we present a linear kernel for the max leaf number parameterization.
In our concluding remarks we propose the further study of the feedback edge set parameter, among other open problems.
\medskip

\noindent {\it Notation and terminology.} 
We denote by $[n]$ the set $\{1,\ldots,n\}$.
All considered graphs are finite and simple.
The \emph{open neighborhood} of a vertex $v\in V(G)$ is $N_G(v) = \{u \in V(G) \mid vu \in E(G)\}$, and its \emph{closed neighborhood} is $N_G[v] = N_G(v) \cup \{v\}$. If the context is clear, we simply write $N(v)$ or $N[v]$. The open neighborhood of $S\subseteq V(G)$ is $N_G(S)=\bigcup_{v \in S}N_G(v)$ and its  closed neighborhood is $N_G[S]=N_G(S)\cup S$.
In an abuse of notation, given $v \in V(G)$ and $S \subseteq V(G)$, we define $N_S(v) = N(v) \cap S$ and $N_S[v] = N[v] \cap S$.
Two distinct vertices $u$ and $v$ of a graph $G$ are {\it false twins} if $N_G(u) = N_G(v)$ and {\it true twins} if  $N_G[u] = N_G[v]$.
Being a false (true) twin is an equivalence relation $\mathcal{R}$, i.e. if $u$ is a false (true) twin of $v$ and $v$ is a false (true) twin of $w$, $u$ is also a false (true) twin of $w$; if $u,v,w$ belong to the same equivalence class of $\mathcal{R}$, we say that they are \textit{mutually twins}.
A graph is \textit{twin-free} if each equivalence class of $\mathcal{R}$ is a singleton.
A graph is said to be a \textit{cluster graph} if each of its connected components is a clique.
The \textit{max leaf number} of a connected graph $G$, denoted by $\ml(G)$, is equal to the maximum number of leaves that a spanning tree of $G$ may have.
As defined in~\cite{cai_split}, a set $U \subseteq V(G)$ is an $\mathcal{F}$-\textit{modulator} of $G$ if $G \setminus U$ belongs to the graph class $\mathcal{F}$, and the \textit{distance to $\mathcal{F}$} of a graph $G$ is the size of a minimum $U$ that is an $\mathcal{F}$-modulator.
A set $S$ is said to be \emph{dominating} if $N_G[S] = V(G)$, while a vertex $v \in V(G)$ is \emph{dominated} by $S$ if $N_S[v] \neq \emptyset$.
The set $S$ is a \emph{locating-dominating set} if it is a dominating set and for every pair of distinct vertices $u,v \in V(G) \setminus S$, $N_S(u) \neq N_S(v)$.
In this case, $u$ and $v$ are said to be \emph{distinguished} by $S$; if $S$ does not distinguish $u$ and $v$, we say that $u,v$ are \emph{confounded} by $S$.

We refer the reader to~\cite{cygan_parameterized} for basic background on parameterized complexity, and recall here only some basic definitions.
A \emph{parameterized problem} is a language $L \subseteq \Sigma^* \times \mathbb{N}$. 
For an instance $I=(x,q) \in \Sigma^* \times \mathbb{N}$, $q$ is called the \emph{parameter}. 
A parameterized problem is \emph{fixed-parameter tractable} (\FPT) if there exists an algorithm $\mathcal{A}$, a computable function $f$, and a constant $c$ such that given an instance $(x,q)$, $\mathcal{A}$ correctly decides whether $I \in L$ in time bounded by $f(q) \cdot |I|^c$; in this case, $\mathcal{A}$ is called an \emph{\FPT\ algorithm}.
A fundamental concept in parameterized complexity is that of \emph{kernelization}; see~\cite{book_kernels} for a recent book on the topic.
A kernelization	algorithm, or just \emph{kernel}, for a parameterized problem $\Pi$ takes an instance~$(x,q)$ of the problem and, in time polynomial in $|x| + q$, outputs an instance~$(x',q')$ such that $|x'|, q' \leqslant g(q)$ for some function~$g$, and $(x,q) \in \Pi$ if and only if $(x',q') \in \Pi$.
Function~$g$ is called the \emph{size} of the kernel and may be viewed as a measure of the ``compressibility'' of a problem using polynomial-time pre-processing rules.
A kernel is called \emph{polynomial} (resp. \emph{quadratic, linear}) if $g(q)$ is a polynomial (resp. quadratic, linear) function in $q$.
A breakthrough result of Bodlaender et al.~\cite{distillation} gave the first framework for proving that some parameterized problems do not admit polynomial kernels, by establishing so-called \emph{composition algorithms}.
Together with a result of Fortnow and Santhanam~\cite{fortnow_santh}, this allows to exclude polynomial kernels under the assumption that $\NP \nsubseteq \coNP/\poly$, otherwise implying a collapse of the polynomial hierarchy to its third level~\cite{uniform_non_uniform}.
\section{Lower bounds for the natural parameter}

While the $d \in \Omega(\log n)$ bound trivially implies that \pname{Locating-Dominating Set} is in \FPT\ when parameterized by $d$, it tells us nothing about the complexity of such an algorithm.
In fact, it could have been the case that a $2^\bigO{d}$ algorithm existed and the hard instances were those where $d \in \omega(\log n)$.
Together with the Sparsification Lemma~\cite{sparsification}, Charon et al.'s~\cite{charon2003} \NPHness\ reduction from \pname{3-SAT} already implied that no such algorithm exists unless the Exponential Time Hypothesis~\cite{eth} fails.
In this section, we go a step further and prove that, unless ETH fails, there is no $2^{o(d\log d)}$ time algorithm for \pname{Locating-Dominating Set}.
As a corollary of our proof, we also prove that no $n^o{d}$ algorithm exists under the same hypothesis.
Our reduction is from the following problem, which does not admit a $2^{o(k\log k)}$ time algorithm~\cite{perm_clique} under ETH:

\probl{$k \times k$ Permutation Clique}{A graph $G$ with vertices labeled $(i,j)$ for every $i,j \in [k]$.}{Is there a permutation $\pi$ of $[k]$ such that $\{(i, \pi(i)) \mid i \in [k]\}$ is a clique of $G$?}

Another way to visualize our source problem is as follows: the vertices of $G$ are arranged in a $k \times k$ grid, and our goal is to find a set of $k$ vertices that induces a clique of $G$ such that we only pick one vertex of each row and each column. W.l.o.g. we assume that each row and column of $G$ is an independent set of $G$.
Throughout this section, we assume that $G$ is the input graph to \pname{$k \times k$ Permutation Clique}, $k$ is the width of the board where $G$ is embedded, and $(H,d)$ is the \pname{Locating-Dominating Set} instance we build.

\smallskip
\noindent\textbf{Construction.} We begin by adding to $H$ a vertex $v_{i,j}$ for each $(i,j) \in V(G)$; let $R(i) = \{v_{i, j} \mid j \in [k]\}$ be the vertex set corresponding to a row of the board and $C(j) = \{v_{i, j} \mid i \in [k]\}$ the vertices corresponding to a column of the board. Now, for each $R(i)$ and $C(j)$ we add a copy of the gadgets $F_{R(i)}$ and $F_{C(j)}$ in Figure~\ref{fig:f_gadgets}; note that $v(i,j)$ is the unique vertex of $H$ that has both $a_{j,1}$ and $\alpha_{i,1}$ in its neighborhood.

\begin{figure}[!htb]
    \centering
        \begin{tikzpicture}[scale=1]
            %\draw[help lines] (-5,-5) grid (5,5);
            \GraphInit[unit=3,vstyle=Normal]
            \SetVertexNormal[Shape=circle, FillColor=black, MinSize=3pt]
            \tikzset{VertexStyle/.append style = {inner sep = \inners, outer sep = \outers}}
            \SetVertexLabelOut
            \begin{scope}
                \Vertex[x=0, y=0.637, Lpos=90,Math, L={a_{j,1}}]{a1}
                \Vertex[x=-1, y=0.637, Lpos=90,Math, L={a_{j,2}}]{a2}
                \Vertex[x=-2, y=0.637, Lpos=90,Math, L={a_{j,3}}]{a3}
                \Vertex[x=-2, y=0.137, Lpos=180,Math, L={a_{j,4}}]{a4}

                \Vertex[x=0, y=-0.637, Lpos=270,Math, L={b_{j,1}}]{b1}
                \Vertex[x=-1, y=-0.637, Lpos=270,Math, L={b_{j,2}}]{b2}
                \Vertex[x=-1, y=-0.137, Lpos=0,Math, L={b_{j,3}}]{b3}
                
                \SetVertexNoLabel
                \Vertex[x=1.25, y=0.95]{h1}
                \Vertex[x=1.25, y=0.32]{h2}
                \Vertex[x=1.25, y=-0.32]{h3}
                \Vertex[x=1.25, y=-0.95]{h4}
                
                \Edges(b3, b2, b1, h1, a1, h2, b1, h3, a1, h4, b1, h1, a1, a2, a3, a4)
                \draw (1, -1.2) rectangle (1.5, 1.2);
                \node at (1.25, -1.5) {$C(j)$};
                
                \draw (a1) circle (0.2cm);
                \draw (a3) circle (0.2cm);
                \draw (b2) circle (0.2cm);
            \end{scope}
            
            \begin{scope}[xshift=2.2cm,yshift=2.2cm,rotate=-90]
                \Vertex[x=0, y=0.637, Lpos=0,Math, L={\alpha_{i,1}}]{a1}
                \Vertex[x=-1, y=0.637, Lpos=90,Math, L={\alpha_{i,2}}]{a2}
                \Vertex[x=-1, y=1.274, Lpos=-90,Math, L={\alpha_{i,3}}]{a3}
                \Vertex[x=-1, y=1.911, Lpos=0,Math, L={\alpha_{i,4}}]{a4}

                \Vertex[x=0, y=-0.637, Lpos=180,Math, L={\beta_{i,1}}]{b1}
                \Vertex[x=-1, y=-0.637, Lpos=90,Math, L={\beta_{i,2}}]{b2}
                \Vertex[x=-1, y=-1.137, Lpos=180,Math, L={\beta_{i,3}}]{b3}
                
                \Vertex[x=0, y=0, Lpos=90,Math, L={p_i}]{p}
                
                \SetVertexNoLabel
                \Vertex[x=1.25, y=0.95]{h1}
                \Vertex[x=1.25, y=0.32]{h2}
                \Vertex[x=1.25, y=-0.32]{h3}
                \Vertex[x=1.25, y=-0.95]{h4}

                \Edges(h1,p,h2,p,h3,p,h4)
                
                \Edges(b3, b2, b1, h1, a1, h2, b1, h3, a1, h4, b1, h1, a1, a2, a3, a4)
                \draw (1, -1.2) rectangle (1.5, 1.2);
                \node at (1.25, 1.6) {$R(i)$};
                
                \draw (a1) circle (0.2cm);
                \draw (a3) circle (0.2cm);
                \draw (b2) circle (0.2cm);
            \end{scope}
        \end{tikzpicture}
\caption{Gadgets $F_{C(j)}$ and $F_{R(i)}$. Circled vertices are present in a canonical solution to $(H,d)$.\label{fig:f_gadgets}}
\end{figure}
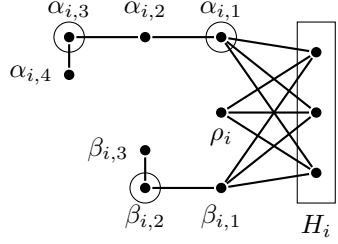

Afterwards, for each $i \in [k]$, we add a $P_4$ $\Sigma_i = \langle\sigma_{i, 1}, \sigma_{i, 2}, \sigma_{i, 3}, \sigma_{i, 4}\rangle$. For each $i_1,i_2 \in [k]$ with $i_1 < i_2$, we add a vertex $\gamma_{i_1, i_2}$ adjacent to $\sigma_{i_1,1}$, $\sigma_{i_2, 1}$, and to all vertices of $R(i_1) \cup R(i_2)$. Now, to encode the adjacencies of $G$, for each $i_1, i_2 \in [k]$ with $i_1 < i_2$, we add the vertex set $R(i_1, i_2) = \{u_{i_1, i_2, j} \mid j \in [k]\}$, add edges between $R(i_2)$ and $R(i_1, i_2)$ to form a perfect matching, make $R(i_1, i_2)$ adjacent to both $\sigma_{i_1, 1}$ and $\sigma_{i_2, 1}$, and add an edge between $v_{i_1, j_1}$ and $u_{i_1, i_2, j_2}$ if and only if $\{(i_1, j_1), (i_2, j_2)\} \notin E(G)$.
Intuitively, the $\gamma$ vertices are used to force that non-adjacent vertices of $H$ correspond to non-adjacent vertices of $G$ cannot be picked simultaneously.
We give an example of this construction in Figure~\ref{fig:edge_gadgets_lb}.
Finally, we add an edge between $v(i_1, j)$ and each $u_(i_2, i_3, j)$ for every $i_1, i_2, i_3 \in [k]$, $i_2 < i_3$, and set $d=9k$.

\begin{figure}[!htb]
    \centering
    
        \begin{subfigure}{0.49\textwidth}
            \centering
            \begin{tikzpicture}[scale=1]
                %\draw[help lines] (-5,-5) grid (5,5);
                \GraphInit[unit=3,vstyle=Normal]
                \SetVertexNormal[Shape=circle, FillColor=black, MinSize=3pt]
                \tikzset{VertexStyle/.append style = {inner sep = \inners, outer sep = \outers}}
                \SetVertexLabelOut            
                \begin{scope}                    
                    \SetVertexNoLabel
                    \Vertex[x=0.95, y=0]{a4}
                    \Vertex[x=0.32, y=0]{a3}
                    \Vertex[x=-0.32, y=0]{a2}
                    \Vertex[x=-0.95, y=0]{a1}
                    
                    \node at (-1.65, 0) {$i_1$};
                    
                    % \draw (a1) circle (0.2cm);
                    % \draw (a3) circle (0.2cm);
                    % \draw (b2) circle (0.2cm);
                \end{scope}
                
                \begin{scope}[yshift=-3cm]
                    \SetVertexNoLabel
                    \Vertex[x=0.95, y=0]{b4}
                    \Vertex[x=0.32, y=0]{b3}
                    \Vertex[x=-0.32, y=0]{b2}
                    \Vertex[x=-0.95, y=0]{b1}
                    
                    \node at (-1.65, 0) {$i_2$};
                    
                    % \draw (a1) circle (0.2cm);
                    % \draw (a3) circle (0.2cm);
                    % \draw (b2) circle (0.2cm);
                \end{scope}
                \Edges(a1,b2,a4,b3,a2,b1,a3,b4)
            \end{tikzpicture}
            \caption{\label{subfig:edge_gadgets_lb_1}}
        \end{subfigure}
        \begin{subfigure}{0.49\textwidth}
            \centering
            \begin{tikzpicture}[scale=1]
                %\draw[help lines] (-5,-5) grid (5,5);
                \GraphInit[unit=3,vstyle=Normal]
                \SetVertexNormal[Shape=circle, FillColor=black, MinSize=3pt]
                \tikzset{VertexStyle/.append style = {inner sep = \inners, outer sep = \outers}}
                \SetVertexLabelOut            
                \begin{scope}                    
                    \SetVertexNoLabel
                    \Vertex[x=0.95, y=0]{a4}
                    \Vertex[x=0.32, y=0]{a3}
                    \Vertex[x=-0.32, y=0]{a2}
                    \Vertex[x=-0.95, y=0]{a1}
                    
                    \draw (-1.2, -0.25) rectangle (1.2, 0.25);
                    \node at (1.65, 0) {$R(i_1)$};
                    
                    % \draw (a1) circle (0.2cm);
                    % \draw (a3) circle (0.2cm);
                    % \draw (b2) circle (0.2cm);
                \end{scope}
                
                \begin{scope}[yshift=-1.6cm]
                    \SetVertexNoLabel
                    \Vertex[x=0.95, y=0]{b4}
                    \Vertex[x=0.32, y=0]{b3}
                    \Vertex[x=-0.32, y=0]{b2}
                    \Vertex[x=-0.95, y=0]{b1}
                    
                    \draw (-1.2, -0.25) rectangle (1.2, 0.25);
                    \node at (1.85, 0) {$R(i_1, i_2)$};
                    
                    % \draw (a1) circle (0.2cm);
                    % \draw (a3) circle (0.2cm);
                    % \draw (b2) circle (0.2cm);
                \end{scope}
                
                \begin{scope}[yshift=-3cm]
                    \SetVertexNoLabel
                    \Vertex[x=0.95, y=0]{c4}
                    \Vertex[x=0.32, y=0]{c3}
                    \Vertex[x=-0.32, y=0]{c2}
                    \Vertex[x=-0.95, y=0]{c1}
                    
                    \draw (-1.2, -0.25) rectangle (1.2, 0.25);
                    \node at (1.65, 0) {$R(i_2)$};
                \end{scope}

                \Edges(a1, b1, c1)
                \Edges(a2, b2, c2)
                \Edges(a3, b3, c3)
                \Edges(a4, b4, c4)
                \Edges(b3,a1,b4)
                \Edges(b4,a2)
                \Edges(b2,a3)
                \Edges(a4,b1)

                \Vertex[x=1.6, y=-0.75, Lpos=0,Math, L={\sigma_{i_1,1}}]{s1}
                \Vertex[x=1.6, y=-2.35, Lpos=0,Math, L={\sigma_{i_2,1}}]{s2}
                \Edges[style={opacity=0.2}](s1,b1,s2,b2,s1,b3,s2,b4,s1)
                
                \Vertex[x=-1.8, y=-1.5, Lpos=180,Math, L={\gamma_{i_1,i_2}}]{g}
                \Edges[style={bend left, opacity=0.2}](g, a1)
                \Edges[style={bend left, opacity=0.2}](g, a2)
                \Edges[style={bend left, opacity=0.2}](g, a3)
                \Edges[style={opacity=0.2}](g, a4)
                \Edges[style={bend left, opacity=0.2}](g, s1)
                
                \Edges[style={bend right, opacity=0.2}](g, c1)
                \Edges[style={bend right, opacity=0.2}](g, c2)
                \Edges[style={bend right, opacity=0.2}](g, c3)
                \Edges[style={opacity=0.2}](g, c4)
                \Edges[style={bend right, opacity=0.2}](g, s2)

                \draw (s1) circle (0.2cm);
                \draw (s2) circle (0.2cm);
            \end{tikzpicture}
            \caption{\label{subfig:edge_gadgets_lb_2}}
        \end{subfigure}
\caption{(a) An example of two rows of the input graph $G$. (b) Corresponding edges between $R(i_1, i_2)$ and other vertices of $H$. Circled vertices are present in a canonical solution to $(H,d)$.\label{fig:edge_gadgets_lb}}
\end{figure}

\begin{lemma}\label{lem:canon_sol}
    In every solution $D$ of $(H,d)$, it holds that (i) $|D \cap F_{X(\ell)}| = 3$ for every $X \in \{R, C\}, \ell \in [k]$, (ii) $|D \cap \Sigma_i| = 2$ for every $i \in [k]$, and (iii) for every $i \in [k]$ $|D \cap R(i)| = 1$ and there is a unique $j \in [k]$ such that $D \cap R(i) = D \cap C(j)$. Furthermore, if $D$ exists, there is also a solution $D^*$, called the canonical solution, where: (iv) $D^* \cap F_{R(i)} = \{\alpha_{i,1}, \alpha_{i,3}, \beta_{i,2}\}$, (v) $D^* \cap F_{C(j)} = \{a_{j,1}, a_{j,3}, b_{j,2}\}$, and (vi) $D^* \cap \Sigma_i = \{\sigma_{i, 1}, \sigma_{i, 3}\}$.
\end{lemma}

\begin{proof}
    Let us show that these values are actually lower bounds for their respective intersections.
    Let us prove do so for gadget $F_{R(i)}$ in property (i); the case for $F_{C(j)}$ is symmetric.
    Note that at least one of $\beta_{i,2}, \beta_{i,3}$ must be in $D$, otherwise $\beta_{i,3}$ would not be dominated; for the $\alpha$ vertices, at least one of $\alpha_{i,3}, \alpha_{i,4}$ must be in $D$, otherwise the latter is not dominated and: (a) if both of them are in $D$, we are done, (b) if only $\alpha_{i,3} \in D$, then one of $\alpha_{i,1}, \alpha_{i,2}$ must also be in $D$, otherwise $\alpha_{i,4}$ would be confounded with $\alpha_{i,2}$, or (c) if only $\alpha_{i,4} \in D$, then one of $\alpha_{i,1}, \alpha_{i,2}$ must also be in $D$, otherwise $\alpha_{i,2}$ is not dominated.
    A similar analysis to the $\alpha$ vertices allows us to conclude 2 is a lower bound for (ii).

    For the first part of property (iii), towards a contradiction, suppose that there is one $\ell \in [k]$ where $D \cap R(\ell) = \emptyset$.
    This implies that $|D \cap R(\ell)| \geq 5$, otherwise $p_i$ would not be dominated (so $p_\ell \in D$) and either $\beta_{\ell,1}$ would not be dominated (if $\beta_{\ell,2} \notin D$) or it would be confounded with $\beta_{\ell,3}$ (if $\beta_{\ell,2} \in D$), so two $\beta_\ell$'s must be in $D$.
    However, note that we now have:
    
    \begin{align*}
        |D| &\geq \sum_{i \in [k]} \left(|D \cap F_R(i)| + |D \cap \Sigma_i| + |D \cap R(i)|\right) + \sum_{j \in [k]} |D \cap F_{C(j)}|\\
           &\geq (k-1)(3+2+1) + \left(|D \cap F_{R(\ell)}| + |D \cap \Sigma_\ell|\right) + 3k\\
           &\geq 6k - 6 + (5 + 2) + 3k\\
           &= 9k + 1 > d
    \end{align*}

    For the second part of property (iii), note that if there was some $q \in [k]$ such that $D \cap C(q) = \emptyset$, then $|\{b_{q,r} \mid r \in [k]\} \cap D| \geq 2$, implying $|D \cap F_{C(j)}| \geq 3$ and, by a similar argument as the above, we would have that $|D| > d$.
    Note that the above lower bounds for (i) - (iii) imply that any valid solution must have $|D| = d$; consequently all three properties hold.

    Let us now prove that (iv) is also true; the cases for (v) and (vi) are symmetric.
    Let $i \in [k]$; if $\beta_{i,3} \in D$, replacing it with $\beta_{i,2}$ does not confound any vertices, as $\beta_{i,1}$ was already dominated by a vertex outside of $F_{R(i)}$ and so $\beta_{i,3}$ is now the unique vertex whose neighborhood in $D$ is $\{\beta_{i,2}\}$.
    For the $\alpha$ vertices, we proceed as follows: (a) if $\alpha_{i,1}, \alpha_{i,3} \notin D$, then $\alpha_{i, 2}, \alpha_{i,4} \in D$ and we can safely replace the latter with the former since $\alpha_{i,2}$ would be the unique vertex dominated by $\{\alpha_{i,1}, \alpha_{i,3}\}$ and $\alpha_{i,4}$ the only one dominated only by $\alpha_{i,3}$; (b) if $\alpha_{i,1} \notin D$ then $\alpha_{i,3} \in D$ and one of $\{\alpha_{i,2}, \alpha_{i,4}\}$ is also in $D$, but then we can safely replace it with $\alpha_{i,1}$ by the same argument as before; (c) if $\alpha_{i, 1} \in D$, then $\alpha_{i,4} \in D$ and we can replace it with $\alpha_{i,3}$ and conclude the correctness of property (iv).
\end{proof}

\begin{lemma}\label{lem:lb_forward}
    If there is a solution $Q$ for $G$, then there is a solution $D$ for $(H,d)$. 
\end{lemma}

\begin{proof}
    Let $Q = \{(1, j_1), \dots, (k, j_k)\}$ be the solution to $G$ and initially $D = \emptyset$. We begin by adding all $8k$ canonical vertices described in Lemma~\ref{lem:canon_sol} to $D$. This way, each vertex $v_{i,j} \in R(i) \cap C(j)$ is uniquely dominated by $\{\alpha_{i,1}, a_{j,1}\}$.
    Moreover, vertices in $R(i_1, i_2) \cup \{\gamma_{i_1, i_2}\}$ are the only ones dominated solely by $\{\sigma_{i_1, 1}, \sigma_{i_2, 1}\}$, so they cannot be confounded with any other vertex of $H$.
    To conclude the construction of $D$, for each $(i,j) \in Q$, we add vertex $v_{i,j}$ to $D$, satisfying $|D| = d = 9k$.
    
    To see that $D$ is a solution to $(H,d)$, we proceed on a per gadget basis.
    First note that the two vertices of $F_{R(i)}$ not uniquely dominated by the canonical vertices are $\beta_{i,1}$ and $p_i$, which are now also dominated by $v_{i, j} \in D$, with $p_i$ being the unique vertex only dominated by $v_{i,j}$ (all other neighbors of $v_{i,j}$ have a canonical vertex in its neighborhood).
    A similar analysis holds for $F_{C(j)}$ and $b_{j,1}$.
    
    For each column $j \in [k]$, define $v_{i,j} \in D \cap C(j)$ as its \textit{dominant}. Recall that, for every $\ell_1, \ell_2 \in [k]$ ($\ell_1 < \ell_2$), $u_{\ell_1, \ell_2, j} \in R(\ell_1, \ell_2)$ is adjacent to $v_{i,j}$. As such, each vertex in $R{\ell_1, \ell_2}$ is dominated by $\sigma_{\ell_1, 1}, \sigma_{\ell_2, 1}$ and its corresponding dominant.
    Note that the only possible vertex that can be confounded with $u_{\ell_1, \ell_2, j}$ at this point is $\gamma_{\ell_1, \ell_2}$ and even then only if the dominant of $j$ is in $R(\ell_1) \cup R(\ell_2)$.
    If the dominant of $j$ is $v_{\ell_1, j} \in D$, then there is some $v_{\ell_2, q} \in D$, which is not adjacent to $u_{\ell_1, \ell_2, j}$, and so there is no confounding with $\gamma_{\ell_1, \ell_2}$.
    Otherwise, $v_{\ell_2, j} \in D$ and there is some $v_{\ell_1, q} \in D$ that is adjacent to $u_{\ell_1, \ell_2, j}$; this, however, cannot be the case, as this edge exists if and only if $\{(\ell_1, q), (\ell_2, j)\}$ is \textbf{not} an edge of $G$ i.e. $\{(\ell_1, q), (\ell_2, j)\}$ is not part of any clique of $G$, a contradiction.
    As such, every vertex in $\bigcup_{i, j, \ell \in [k]} R(i) \cup C(j) \cup F_{R(i)} \cup F_{C(j)} \cup \Sigma_i \cup R(i, \ell) \cup \{\gamma_{i, \ell}\} = H$ is either in $D$ or uniquely dominated by a subset of $D$.
\end{proof}

\begin{lemma}\label{lem:lb_backward}
    If there is a solution $D$ for $(H,d)$, then there is a solution $Q$ for $G$. 
\end{lemma}

\begin{proof}
    W.l.o.g. let us assume that $D$ is a canonical solution as laid out in Lemma~\ref{lem:canon_sol} and initially $Q = \emptyset$. 
    For each $j \in [k]$, we add the dominant $v_{i,j}$ of column $j$ to $Q$.
    Since $D$ is canonical, the dominants are unique and each $R(i)$ has exactly one element in $Q$, so $|Q| = k$, every row of $G$ has a unique vertex in $Q$, and so does each column.
    It remains to prove that $Q$ is a clique.
    Suppose that this is not the case, i.e. there are two elements $(i_1, j_1), (i_2, j_2) \in Q$ ($i_1 < i_2$) that are not adjacent.
    As such, it holds that $v_{i_1, j_1}, v_{i_2, j_2}\in D$.
    Moreover, since $D$ is canonical, $\gamma_{i_1, i_2} \notin D$ and $N_D(\gamma_{i_1, i_2}) = \{v_{i_1, j_1}, v_{i_2, j_2}, \sigma_{i_1, 1}, \sigma_{i_2, 1}\}$.
    However, by the construction of $H$, the vertex $u_{i_1, i_2, j_2} \in R(i_1, i_2) \setminus D$ is adjacent all vertices in $N_D(\gamma_{i_1, i_2})$ and since the dominant of $j_2$ is $v_{i_2, j_2}$ and no other vertex of $R(i_1)$ is in $D$, no other vertex of $H$ is in $D$, and consequently $u_{i_1, i_2, j_2}$ and $\gamma_{i_1, i_2}$ are confounded, a contradiction.
    As such, we conclude that $Q$ is a permutation clique of $G$.
\end{proof}

\begin{theorem}\label{thm:lower_bound_dlod}
    \pname{Locating-Dominating Set} does not admit an $2^{o(d\log d)}\poly(n)$ time algorithm unless the Exponential Time Hypothesis fails.
\end{theorem}

\begin{proof}
    Lemmas~\ref{lem:lb_forward} and~\ref{lem:lb_backward} immediately imply that $G$ and $(H,d)$ are equivalent. Together with the fact that $d = 9k$ in our construction, the result of Lokshtanov et al.~\cite{perm_clique} that \pname{$k \times k$ Permutation Clique} does not admit a $2^{o(k\log k)}$ time algorithm unless the Exponential Time Hypothesis fails immediately implies that \pname{Locating-Dominating Set} does not admit an $2^{o(d\log d)}\poly(n)$ time algorithm under the same hypothesis.
\end{proof}

\begin{corollary}
    Unless the Exponential Time Hypothesis fails, there is no $n^{o(d)}$ algorithm for \pname{Locating-Dominating Set} on $n$-vertex graphs.
\end{corollary}

\begin{proof}
    Suppose that there is such such an algorithm.
    By our construction, we have that $n = |V(H)| = \Theta(k^3)$ and $d = \Theta(k)$, so it holds that $n^{o(d)} = k^{3o(k)} = 2^{o(k) \log k} = 2^{o(k\log k)}$ and we would have an algorithm with the aforementioned running time for \pname{$k \times k$ Permutation Clique}.
\end{proof}
\section{An exponential kernel for distance to cluster}

Throughout this section, let $(G, d)$ be the input to \pname{Locating-Dominating Set}, $U$ be a cluster modulator of $G$, and define $V(G) \setminus U = \mathcal{C}$.
A clique $Q \in \mathcal{C}$ is \textit{trivial} if no two vertices of $Q$ are true twins to one another.
We say that two maximal cliques $Q_1, Q_2 \in \mathcal{C}$ are \textit{twins} if there is a bijection $f : V(Q_1) \mapsto V(Q_2)$ with $f(u) = v$ if and only if $N_U(u) = N_U(v)$; a family of cliques $\mathcal{Q} \subseteq \mathcal{C}$ is a \textit{pattern} if all of its cliques are twins and no other clique of $\mathcal{C} \setminus \mathcal{Q}$ is a twin of a clique of $\mathcal{Q}$.
We say that a pattern $\mathcal{Q}$ is a pattern of \textit{trivial cliques} if, for every $Q \in \mathcal{Q}$, $Q$ is a trivial clique.
For simplicity, if $Q_i, Q_j \in \mathcal{Q}$, we denote $Q_i = \{v_1^i, \dots, v_{|Q_i|}^i\}$ and define that $N_U(v_\ell^i) = N_U(v_\ell^j)$ for every $\ell \in \left[|Q_i|\right]$; finally, we define $V(\mathcal{Q}) = \bigcup_{Q \in \mathcal{Q}} V(Q)$.

Our first reduction rule is a well known pre-processing step for \pname{Locating-Dominating Set}.

\begin{reduction}
    \label{rule:twinning}
    If $u,v,w \in V(G) \setminus U$ is a set of mutual twins of $G$, remove $w$ from $G$ and set $d \gets d - 1$.
\end{reduction}

Before our next reduction rule, we present a lemma that allows us to normalize solutions to more easily prove the safeness of the rule.

\begin{lemma}
    \label{lem:trivial_cliques}
    Let $D$ be a solution to $(G, d)$ and $\mathcal{Q}$ be a pattern of $r$ trivial cliques, each of size $s \geq 2$.
    If $r \geq 2s + |U| + 1$, then there is an index $\ell \in [s]$, three cliques $Q_a, Q_b, Q_c \in \mathcal{Q}$, and a solution $D^*$ of size at most $|D|$ where, for every $j \in \{a,b,c\}$, $D^* \cap Q_j = \{v_\ell^j\}$.
\end{lemma}

\begin{proof}
    If $|D \cap V(\mathcal{Q})| \geq r + |U|$, there are at least $|U|$ cliques that contain two vertices of $D$.
    As such, let $D^* = (D \cup U \setminus V(\mathcal{Q})) \cup \{v_1^i \in Q_i \mid Q_i \in \mathcal{Q}\}$; we claim that $D^*$ is a solution of size at most $|D|$.
    To see that this is the case, note that $D \cap U \subseteq D^* \cup U$ and, by adding vertices to $U \setminus D$, vertices already distinguished by $D \cap U$ remain distinguished by $D^* \cap U = U$; moreover no vertex may be confounded with a vertex of a clique $Q_i \in \mathcal{Q}$ since $Q_i$ has a vertex in $D^*$, which distinguishes $Q_i$ from every other clique of $\mathcal{C}$, which implies that all vertices of $G$ are distinguished by $D^*$.
    As to the size of $D^*$, note that we add at most $r + |U|$ vertices to $D \setminus V(\mathcal{Q})$, i.e. $|D^*| \leq r + |U| + |D \setminus V(\mathcal{Q})| = r + |U| + |D| - |D \cap V(\mathcal{Q})| \leq |D|$.
    Since $r > 2s + |U|$ and each clique of $\mathcal{Q}$ has one vertex in $D^*$ and $s$ vertices in total, by the pigeonhole principle there is some index $\ell \in [s]$ and three cliques $Q_a, Q_b, Q_c \in \mathcal{Q}$ such that $v_\ell^j \in D^*$ for every $j \in \{a,b,c\}$.
    
    On the other hand, if $|D \cap V(\mathcal{Q})| \leq r + |U| - 1$, we observe two properties: (i) at most one clique of $\mathcal{Q}$ has no vertex in $D$, and (ii) at most $|U| - 1$ cliques have more than one vertex in $D$.
    Putting these together, we have that $r - 1 - (|U| - 1) = r - |U| \geq 2s + 1$ cliques have exactly one vertex in $D$; by the pigeonhole principle, at least three cliques $Q_a, Q_b, Q_c \in \mathcal{Q}$ have $v_\ell^j \in D$ for every $j \in \{a,b,c\}$.
    Set $D^* = D$.
\end{proof}

\begin{reduction}
    \label{rule:trivial_cliques}
    Let $\mathcal{Q}$ be a pattern of $r$ trivial cliques, each of size $s \geq 2$.
    If $r \geq 2s + |U| + 2$, remove a clique from $\mathcal{Q}$ and set $d \gets d - 1$.
\end{reduction}

\begin{proof}
    Let $D$ be a solution to $(G, d)$ and $D^*$ the solution obtained after applying Lemma~\ref{lem:trivial_cliques}.
    For the forward direction, we claim that $D' = D^* \setminus Q_a$ is a solution for $(G', d-1)$, where $G' = G \setminus Q_a$.
    Suppose to the contrary, that some pair of vertices $x,y \in V(G')$ is confounded by $D'$ in $G'$ but not by $D^*$ in $G$.
    Since $D' \cap Q_a = \{v_\ell^a\}$, w.l.o.g. it holds that $N_{D^*}(x) = N_{D'}(x) \cup \{v_\ell^a\}$; since $x \notin Q_a$, it must be the case that $x \in U \setminus D$, so we conclude that $v_\ell^b, v_\ell^c \in N_{D'}(x)$; moreover, since $D^*$ is a solution for $(G, d)$, we have that $v_\ell^a \notin N_{D^*}(y)$, otherwise $N_{D^*}(y) = N_{D^*}(x)$.
    Thus, for $N_{D'}(x)$ to be equal to $N_{D'}(y)$, it must be the case that $v_\ell^b, v_\ell^c \in N_{D'}(y)$, so $y \in U \setminus D'$; we point out that this is where we use the fact that we have three cliques with similar intersection with $D$.
    However, since $yv_\ell^a \notin E(G)$, we have $N_U(v_\ell^a) \neq N_U(v_\ell^b)$, a contradiction to the hypothesis that $Q_a$ and $Q_b$ belong to the same pattern $\mathcal{Q}$.
    
    For the converse, let $Q_i$ be the clique removed from $G$, $G' = G \setminus Q_i$, and note that $|\mathcal{Q} \setminus \{Q_i\}| \geq 2r + |U| + 1$; that is, given a solution $D'$ to $(G', d-1)$, there is a solution $D^*$ with three cliques $Q_a, Q_b, Q_c \in \mathcal{Q}$ and an index $\ell \in [s]$ where $v_\ell^j \in D^*$ for every $j \in \{a,b,c\}$.
    As such, we define $D = D^* \cup \{v_\ell^i\}$, where $v_\ell^i \in Q_i$, and claim that it is a solution to $(G, d)$.
    Since $Q_i \cap D \neq \emptyset$, no vertex of $Q_i$ may be confounded with a vertex of another clique of $\mathcal{C}$ and, since $D^* \subset D$ is a solution to $(G', d-1)$, no two vertices of $G'$ may now be confounded with each other.
    We are left with two cases:
    
    \begin{enumerate}
        \item There is a vertex $v_t^i \in Q_i$ confounded with a vertex $y \in U \setminus D$: if this is true, then it must be the case that $yv_\ell^i \in E(G)$, which implies $yv_\ell^a \in E(G)$, but $v_t^iv_\ell^a \notin E(G)$, so $N_{D}(v_t^i) \neq N_{D}(y)$.
        \item There are two vertices $v_t^i, v_p^i \in Q_i$ that are confounded: this cannot hold, since $N_{D}(v_t^i) \setminus \{v_\ell^i\} = N_{D}(v_t^a)~\setminus~\{v_\ell^a\} \neq N_{D}(v_p^a)~\setminus~\{v_\ell^a\} = N_{D}(v_p^i) \setminus \{v_\ell^i\}$, otherwise $D^*$ would not be a solution to $(G', d-1)$.
    \end{enumerate}
    
    Consequently, $D$ is a solution to $(G, d)$, and, finally, $|D| = |D^*| + 1 \leq d$.
\end{proof}

Once Rule~\ref{rule:trivial_cliques} can no longer be applied, there is a bounded number of trivial cliques in $\mathcal{C}$.
To bound the non-trivial cliques, we heavily rely on the fact that each of these cliques must have at least one vertex in the solution, which allows us to obtain significantly better bounds than in Rule~\ref{rule:trivial_cliques}; recall that, by Rule~\ref{rule:twinning}, each vertex of a non-trivial clique $Q$ has at most one twin in $Q$, so the size of $Q$ is also bounded by a function on $|U|$.
We denote by $\tau(Q)$ the set of vertices of $Q$ that contains precisely one vertex of each true twin class of size two of $Q$; in an abuse of notation, we define $\tau(\mathcal{Q}) = \tau(Q)$ for an arbitrary $Q \in \mathcal{Q}$.

\begin{lemma}
    \label{lem:other_cliques}
    Let $D$ be a solution to $(G,d)$ and $\mathcal{Q}$ be a pattern of $r$ non-trivial cliques, each of size $s \geq 2$.
    If $r \geq |U| + 1$, there is a solution $D^*$ of size at most $|D|$ and a clique $Q_i \in \mathcal{Q}$ where $D^* \cap Q_i = \tau(Q_i)$.
\end{lemma}

\begin{proof}
    If $Q_i$ satisfies $D \cap Q_i = \tau(Q_i)$, we are done by setting $D^* = D$.
    Otherwise, it holds that $|D \cap V(\mathcal{Q})| \geq r(|\tau(\mathcal{Q})|+1)$ and we proceed similarly to Lemma~\ref{lem:trivial_cliques}: since we have at least $|U|$ cliques of $\mathcal{Q}$ with $|\tau(\mathcal{Q})| + 1$ vertices in the solution, we set $D^* = D \cup U \setminus \bigcup_{Q \in \mathcal{Q}} (Q \setminus \tau(Q))$, which, by our assumption, has at most $|D|$ vertices.
    Since $D$ is a solution and the $|\mathcal{Q}| \geq 2$ cliques of $\mathcal{Q}$ have all of their $|\tau(\mathcal{Q})|$ twins in $D^*$, the latter is also a solution to $(G,d)$.
\end{proof}

\begin{reduction}
    \label{rule:other_cliques}
    Let $\mathcal{Q}$ be a pattern of $r$ non-trivial cliques, each of size $s \geq 2$.
    If $r \geq |U| + 2$, remove a clique from $\mathcal{Q}$ and set $d \gets d - |\tau(\mathcal{Q})|$.
\end{reduction}

\begin{proof}
    Let $D$ be a solution to $(G, d)$, and $D^*$ the solution obtained using Lemma~\ref{lem:other_cliques}.
    If $Q_i$ is a clique with $D^* \cap Q_i = \tau(Q_i)$, $D' = D^* \setminus Q_i$, and $G' = G \setminus Q_i$, we claim that $D'$ is a solution to $(G', d - |\tau(\mathcal{Q})|$.
    Suppose to the contrary, let $x,y \in V(G')$ be two confounded vertices, i.e. $N_{D'}(x) = N_{D'}(y)$, and w.l.o.g assume there is some $u_\ell^i \in N_{D^*}(x) \cap Q_i \subseteq \tau(Q_i) $, that is $x \in U \setminus D^*$ and $X$ is adjacent to every $u_\ell^j \in \tau(Q_j)$, for $Q_j \in \mathcal{Q}$; this implies, however, that $y$ is adjacent to $u_\ell^j$ for every $Q_j \in \mathcal{Q} \setminus \{Q_i\}$, but not to $u_\ell^i$, a contradiction to the hypothesis that $Q_i,Q_j$ are twin cliques.
    
    For the converse, let $D'$ be a solution to $(G', d - |\tau(\mathcal{Q})|)$, with $G' = G \setminus Q_i$, $D^*$ the solution to $(G', d - |\tau(\mathcal{Q})|$ given by Lemma~\ref{lem:other_cliques}, and $Q_j \in \mathcal{Q}$ be such that $D^* \cap Q_j = \tau(Q_j)$.
    To see that $D = D^* \cup \tau(Q_i)$ is a solution to $(G,d)$, note that: (i) no two vertices in $V(G')$ be confounded, because $D^* \subseteq D$ is a solution to $(G', d - |\tau(\mathcal{Q})|)$, and (ii) no vertex of $v_\ell^i \in Q_i$ may be confounded with another vertex of $G$, otherwise $v_\ell^j$ would be confounded with the same vertex.
\end{proof}

\begin{theorem}
    When parameterized by the distance to cluster $k$, \pname{Locating-Dominating Set} admits a kernel with $\bigO{2^{8^k + 3k}}$ vertices and can be computed in $\bigO{m + n\log n}$ time on graphs with $n$ vertices and $m$ edges.
\end{theorem}

\begin{proof}
    The first step is to obtain a distance to cluster modulator $U$ in $\bigO{n+m}$ time; this is done by noting that a graph is a cluster graph if and only if it is $P_3$-free.
    As such, we can remove, for every induced $P_3$ of $G$, all three vertices; since at least one of these is in a modulator of minimum size, the resulting modulator $U$ has at most $3k$ vertices.
    
    With $U$ in hand, if $n \leq 2^{2^{|U|}}$, we are done.
    So now, we begin by applying Rule~\ref{rule:twinning}, which can be done in $\bigO{n+m}$ time: we keep a table $T(S, i)$ of integers, where $S$ is one of the $2^{|U|} < n$ subsets of $U$ and $i \in [|V(\mathcal{C})|]$ is an index of a clique of $\mathcal{C}$, that stores how many vertices of clique $Q_i \in \mathcal{C}$ have $S$ as its neighborhood in $U$; this table may be kept as a prefix tree, and each entry is initialized to zero.
    To fill $T$, for each $Q_i \in \mathcal{C}$ and $v \in Q_i$, we access $T(N_U(v), i)$ in $\bigO{|N_U(v)|}$ time and check: if $T(N_U(v), i) \geq 2$, delete $v$ from $G$, otherwise, increase the entry by one unit and keep $v$ in $G$; a straightforward analysis shows that this process takes $\bigO{\sum_{v \in V(\mathcal{C})} |N_U(v)|} = \bigO{m}$ time.
    
    For the remaining rules, the costly step is computing the $\bigO{2^{2^{|U|}}}$ patterns of $\mathcal{C}$.
    It is not hard to see that we can partition $\mathcal{C}$ in $\{\mathcal{Q}_1, \dots, \mathcal{Q}_t\}$ patterns in $\bigO{m + n\log n}$ time using the fact that $n > 2^{2^{|U|}}$ and some prefix trees: it suffices to sort the cliques by the neighborhood of their vertices in $U$ and insert these lists of sorted vertices in a prefix tree $F$ one at a time; cliques that fall on the same node of $F$ have the same pattern.
    Now, for each pattern $\mathcal{Q}_i$, where $i \in [t]$, we can check if it is a trivial or non-trivial pattern by picking any clique and check how many pairs of twin vertices it contains, which, when naively done, takes $\bigO{2^{2|U|}} = \bigO{\log^2 n}$ time.
    Since we know the size of $\mathcal{Q}_i$ and $|\tau(\mathcal{Q}_i)|$, and Rules~\ref{rule:trivial_cliques} and~\ref{rule:other_cliques} consist only of removing a clique, the applications of these rules to all patterns amount to $\bigO{n}$ time.
    
    As to the size of the reduced graph $G'$, note that we have at most $2\cdot2^{2^{|U|}}$ patterns in $\mathcal{C}$, each clique has at most $2^{|U|+1}$ vertices and, by Rule~\ref{rule:trivial_cliques}, each pattern has at most $\bigO{2^{|U|}}$ cliques.
    Putting these together, $G'$ has at most $\bigO{2^{2^{|U|} + |U|}} = \bigO{2^{8^k + 3k}}$ vertices, concluding the proof.
\end{proof}

Using the fact that a 2-approximation for the size of a minimum distance to clique modulator can be found in $\bigO{n + m}$ time, we obtain the following corollary by a simple application of Rule~\ref{rule:twinning}.

\begin{corollary}
    When parameterized by the distance to clique $k$, \pname{Locating-Dominating Set} admits a kernel with $\bigO{16^k}$ vertices and can be computed in $\bigO{m + n}$ time on graphs with $n$ vertices and $m$ edges.
\end{corollary}
\section{Lower bounds for kernelization}
\label{sec:nokernel}

\subsection{Vertex Cover}

We are going to show that the \NPH\ problem \textsc{3-Uniform Hypergraph Bicoloring}~\cite{lovasz_hypergraph} OR-cross-composes~\cite{cross_composition} into \textsc{Locating-Dominating Set} parameterized by vertex cover and size of the solution.
To this end, let $\mathcal{H} = \{H_0, \dots, H_t\}$ be the input instances to \textsc{3-Uniform Hypergraph Bicoloring}, with $H_i = (U, E_i)$, all of which have $n$ vertices, and $(G, d)$ the instance of  \textsc{Locating-Dominating Set} we wish to build.
Moreover, we assume w.l.o.g. that $|\mathcal{H}|$ is a power of two.
Finally, we define $\mathcal{E} = \bigcup_{0 \leq i \leq t} E_i$ and $m = |\mathcal{E}|$.

\smallskip
\noindent\textbf{Construction.}
For each $v \in U$, add a copy $G_v(\alpha, \beta)$ of the gadget in Figure~\ref{fig:vertex_gadget} to $G$, which was first used by Charon et al.~\cite{charon2003}.
Now, for each hyperedge $\varepsilon \in \mathcal{E}$, with $\varepsilon = \{u,v,w\}$, add a copy $G_\varepsilon$ of the gadget of Figure~\ref{fig:edge_gadget}; vertex $x_i$ belongs to the instance selector gadget, which we describe below.
Note also that the vertices adjacent to $A_\varepsilon$ and $B_\varepsilon$ are vertices of the corresponding vertex gadget.

\begin{figure}[!htb]
    \centering
        \begin{tikzpicture}[scale=1]
            %\draw[help lines] (-5,-5) grid (5,5);
            \GraphInit[unit=3,vstyle=Normal]
            \SetVertexNormal[Shape=circle, FillColor=black, MinSize=3pt]
            \tikzset{VertexStyle/.append style = {inner sep = \inners, outer sep = \outers}}
            \SetVertexLabelOut
            \Vertex[x=-3, y=0, Lpos=180,Math]{a_v}
            \Vertex[x=-2, y=0, Lpos=90,Math]{b_v}
            \Vertex[x=-1, y=0, Lpos=0,Math]{c_v}
            \Vertex[x=1, y=0, Lpos=180,Math]{d_v}
            \Vertex[x=2, y=0, Lpos=0,Math]{e_v}
            
            \Vertex[x=0, y=1, Lpos=90,Math, L={\alpha_v}]{al_v}
            \Vertex[x=0, y=-1, Lpos=270,Math, L={\beta_v}]{be_v}
            
            \Edges(b_v,c_v,al_v,d_v,be_v,c_v)
            \Edges(d_v, e_v)
            \Edge[style={bend left=40}](a_v)(al_v)
            \Edge[style={bend right=40}](a_v)(be_v)
        \end{tikzpicture}
\caption{Binary choice gadget $G_v(\alpha, \beta)$.\label{fig:vertex_gadget}}
\end{figure}
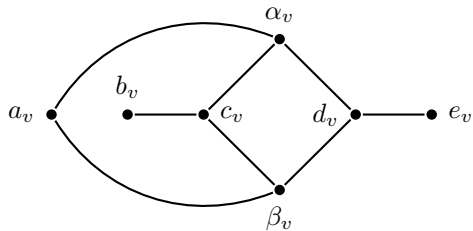

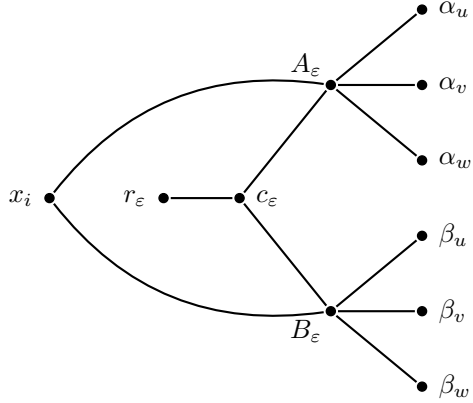
\begin{figure}[!htb]
    \centering
        \begin{tikzpicture}[scale=1]
            %\draw[help lines] (-5,-5) grid (5,5);
            \GraphInit[unit=3,vstyle=Normal]
            \SetVertexNormal[Shape=circle, FillColor=black, MinSize=3pt]
            \tikzset{VertexStyle/.append style = {inner sep = \inners, outer sep = \outers}}
            \SetVertexLabelOut
            \Vertex[x=-4.5, y=0, Lpos=180,Math]{x_i}
            \Vertex[x=-3, y=0, Lpos=180,Math, L={r_\varepsilon}]{r_e}
            \Vertex[x=-2, y=0, Lpos=0,Math, L={c_\varepsilon}]{c_e}
            \Vertex[x=-0.8, y=1.5, Lpos=135, Ldist=-3pt,Math, L={A_\varepsilon}]{A_e}
            \Vertex[x=-0.8, y=-1.5, Lpos=225, Ldist=-3pt,Math, L={B_\varepsilon}]{B_e}
            
            \Vertex[x=0.4, y=2.5, Lpos=0,Math, L={\alpha_u}]{al_u}
            \Vertex[x=0.4, y=1.5, Lpos=0,Math, L={\alpha_v}]{al_v}
            \Vertex[x=0.4, y=0.5, Lpos=0,Math, L={\alpha_w}]{al_w}
            
            \Vertex[x=0.4, y=-0.5, Lpos=0,Math, L={\beta_u}]{be_u}
            \Vertex[x=0.4, y=-1.5, Lpos=0,Math, L={\beta_v}]{be_v}
            \Vertex[x=0.4, y=-2.5, Lpos=0,Math, L={\beta_w}]{be_w}
            
            \Edges(r_e,c_e,A_e,al_u)
            \Edges(c_e,B_e,be_u)
            \Edges(be_v,B_e,be_w)
            \Edges(al_v,A_e,al_w)
            \Edges[style={bend right}](A_e, x_i, B_e)
            
        \end{tikzpicture}
\caption{Hyperedge gadget for some $\varepsilon \in \mathcal{E}$ with $\varepsilon = \{u,v,w\}$.\label{fig:edge_gadget}}
\end{figure}

The more involved piece of our construction is our instance selector gadget $T$.
We first add $h = \log_2 |\mathcal{H}|$ copies $\mathcal{T} = \{T_0(0,1), \dots, T_{h-1}(0,1)\}$ of the binary choice gadget to $T$, then add $|\mathcal{H}|$ vertices $X = \{x_0, \dots, x_t\}$, and five auxiliary vertices $Y = \{y_0, y_1, y_2, z, z'\}$ to the gadget; the auxiliary vertices $Y \setminus \{y_2\}$ form a claw with $z$ as its center.
Then, for each $x_i \in X$ and $T_j \in \mathcal{T}$, if the $j$-th bit of $i$ is equal to zero, add the edge $x_i1_j$, otherwise add edge $x_i0_j$.
Now, add all edges between vertices in $Y \setminus \{z'\}$ and $X$, and, for each $T_j \in \mathcal{T}$, add the edges $y_00_j$ and $y_01_j$.
Finally, for each hyperedge $\varepsilon \in \mathcal{E}$ and $H_i \in \mathcal{H}$, if $\varepsilon \notin E_i$, add the edges $x_iA_\varepsilon$ and $x_iB_\varepsilon$.
We present an example for $h=2$ in Figure~\ref{fig:sel_gadget}.
Intuitively, vertices $0_j$ and $1_j$ describe a bitmask that uniquely distinguishes the vertices of $X$ from each other, but are not enough to: (i) dominate all elements of $X$ and (ii) distinguish between one vertex of $X$ and $y_0$.
To conclude the construction of our instance, we set $d = 3(n + h) + m + 2$.

\begin{figure}[!htb]
    \centering
        \begin{tikzpicture}[scale=1]
            %\draw[help lines] (-5,-5) grid (5,5);
            \GraphInit[unit=3,vstyle=Normal]
            \SetVertexNormal[Shape=circle, FillColor=black, MinSize=3pt]
            \tikzset{VertexStyle/.append style = {inner sep = \inners, outer sep = \outers}}
            \SetVertexLabelOut
            \Vertex[x=-1, y=1.5, Lpos=90,Math]{x_0}
            \Vertex[x=1, y=1.5, Lpos=90,Math]{x_3}
            \Vertex[x=-1, y=-1.5, Lpos=270,Math]{x_1}
            \Vertex[x=1, y=-1.5, Lpos=270,Math]{x_2}
            
            \Vertex[x=-3, y=0, Lpos=180,Math]{y_1}
            \Vertex[x=0, y=0, Lpos=270, Ldist=1pt,Math, L={y_0}]{y}
            \Vertex[x=3, y=0, Lpos=0,Math]{y_2}
            
            \Vertex[x=0, y=3, Lpos=90, Math]{z}
            \Vertex[x=2, y=3, Lpos=0, Math, L={z'}]{zp}
            
            \Edge(z)(y)
            \Edge[style={bend left=1}](z)(zp)
            \Edge[style={bend right=15}](z)(y_1)
            
            \draw[line width=0.8pt] plot [smooth] coordinates {(0, 3) (-2, 2) (-4, 0) (-1, -1.5)};
            \draw[line width=0.8pt] plot [smooth] coordinates {(0, 3) (2, 2) (4, 0) (1, -1.5)};
            
            \Edges[style={bend right}](x_0,z,x_3)
            
            \Edges[style={opacity=0.2}](y,x_0,y_1, x_1, y, x_2, y_1, x_3, y)
            \Edges[style={opacity=0.2}](x_0,y_2,x_1)
            \Edges[style={opacity=0.2}](x_2,y_2,x_3)
            
            \SetVertexNoLabel
            \Vertex[x=-1.725, y=0, Lpos=0,Math]{al_0}
            \Vertex[x=-1.275, y=0, Lpos=0,Math]{be_0}
            \draw (-1.10, -0.25) rectangle (-1.9, 0.25);
            
            \Vertex[x=1.275, y=0, Lpos=0,Math]{al_1}
            \Vertex[x=1.725, y=0, Lpos=0,Math]{be_1}
            \draw (1.10, -0.25) rectangle (1.9, 0.25);
           
            \Edges(be_0, y)
            \Edge[style={bend left}](y)(al_0)
            \Edge[style={bend right}](y)(be_1)
            \Edge(y)(al_1)
            \Edges[style={opacity=0.2}](be_0, x_0, be_1)
            \Edges[style={opacity=0.2}](be_0, x_1, al_1)
            \Edges[style={opacity=0.2}](al_0, x_2, be_1)
            \Edges[style={opacity=0.2}](al_0, x_3, al_1)
            \begin{scope}
                \tikzset{VertexStyle/.append style = {shape = rectangle, inner sep = 2.2pt}}
                \AddVertexColor{black}{be_0, be_1}
                \AddVertexColor{white}{al_0, al_1}
            \end{scope}
            
        \end{tikzpicture}
\caption{Instance selector gadget for $h = 2$. Black rectangle are vertices of the form $1_j$, while white rectangles the $0_j$ vertices.\label{fig:sel_gadget}}
\end{figure}
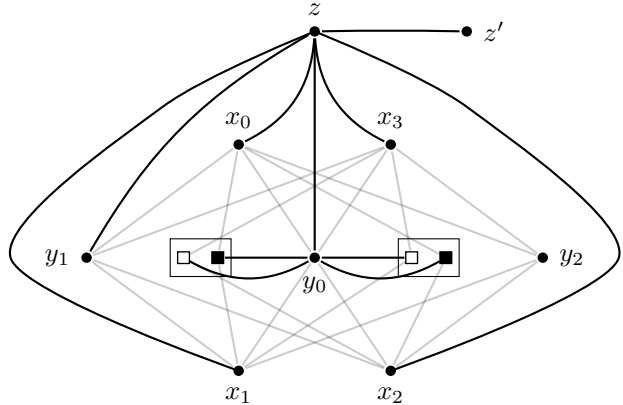

%Our vertex gadgets are the same ones used by Charon et al.~\textbf{(CITACAO PARA "MINIMIZING...")} in their proof that \textsc{Locating-Dominating Set} is NP-complete, while the hyperedge gadgets are inspired by the gadget of Gomes and Sau in~\textbf{(CITAR D-CUT)}.

To prove the correctness of our solution, we employ Lemma~\ref{lem:fw_no_kernel} to construct a solution to $(G,d)$ from a solution to some $H_i \in \mathcal{H}$.
The converse, however, requires a bit more of work: we first show that every solution to $(G,d)$ must have a specific form; this canonical form allows us to easily identify which instance $H_i \in \mathcal{H}$ has a solution and to properly reconstruct it in polynomial time.

\begin{lemma}
    \label{lem:fw_no_kernel}
    If some $H_i \in \mathcal{H}$ has a solution, then $(G,d)$ admits a locating code with at most $d$ vertices.
\end{lemma}

\begin{proof}
    Let $\varphi$ be a bicoloring of $U = V(H_i)$ with colors $\alpha$ and $\beta$.
    For every, $v \in U$, we add to the solution $S$ the set $\{c_v, d_v\}$ and, if $v \in \varphi_\alpha$, we add $\alpha_v$ to $S$, otherwise we add $\beta_v$ to $S$; at this point $|S| = 3n$.
    For every $\varepsilon \in \mathcal{E}$, we add $c_\varepsilon$ to $S$, setting $|S| = 3n + m$.
    
    \begin{claim}
        \label{clm:fw_sub}
        $S$ is a solution for the subgraph $G^*$ of $G$ induced by all the vertex gadgets and the hyperedge gadgets for hyperedges in $E(H_i)$.
    \end{claim}
    
    \begin{cproof}{\ref{clm:fw_sub}}
        For $v \in U$, let $\gamma_v = \{\alpha_v, \beta_v\} \setminus S$, which has $N_S(\gamma_v) = \{c_v,d_v\}$ and, since no other vertex in $V(G^*) \setminus S$ has both $c_v$ and $d_v$ in the neighborhood, $\gamma_v$ cannot be confounded with another vertex of $G^*$.
        Now, since $\varphi$ is a solution to $H_i$, it holds that, in each vertex gadget $G_\varepsilon$ for $\varepsilon \in E(H_i)$, both $A_\varepsilon$ and $B_\varepsilon$ have neighbors in $S \setminus V(G_\varepsilon)$, otherwise $\varepsilon$ would be a monochromatic hyperedge.
        Consequently, since $c_\varepsilon \in S$ and no vertex outside of $G_\varepsilon$ has $c_\varepsilon$ in its neighborhood, $A_\varepsilon$, $B_\varepsilon$, and $r_\varepsilon$ cannot be confounded with other vertices; this in turn implies that $|N_S(A_\varepsilon)|, |N_S(B_\varepsilon)|\geq 2$ for every $\varepsilon \in E(H_i)$, so it also holds that, for each $v \in U$, no other vertex of
        $G^*$ has only $\alpha_v$ or $\beta_v$ in its neighborhood and in $S$, so $a_v$ cannot be confounded with another vertex.
        Thus, every vertex of $G^*$ is uniquely identified by their neighborhood in $S$, proving the claim.
    \end{cproof}
    
    Now, for each $T_j \in \mathcal{T}$, we add to $S$ vertices $\{c_j, d_j\}$ and, if the $j$-th bit of $i$ is equal to 0, add $0_j$, otherwise add $1_j$, i.e we add to $S$ the vertex in $T_j$ that is \textbf{not} adjacent to $x_i$; this increases the size of $S$ to $3(n + h) + m$.
    Finally, add $x_i$ and $z$ to $S$, meeting the bound $d$.
    Before analyzing the situation in $T$, we point out that, for each $\varepsilon \notin E(H_i)$, vertices $A_\varepsilon$ and $B_\varepsilon$ are dominated by $x_i$ and $c_\varepsilon$, while $r_\varepsilon$ is only dominated by the latter, so no vertex in $G_\varepsilon$ can be confounded with another outside of the gadget, regardless of how the $\alpha$'s and $\beta$'s of $G_\varepsilon$ intersect $S$.
    
    To see that $S$ is a solution to $G$, first note that every vertex in $V(T) \setminus \{y_2\}$ has at least one neighbor in $S \cap T \setminus \{x_i\}$, so they cannot be confounded with a vertex outside of $T$.
    Note also that the bitmask described by the vertices of $\mathcal{T}$ is precisely the binary representation of $x_i$.
    By the way we defined the edges between vertices of $X$ and $\mathcal{T}$, each $x_\ell \neq x_i$ is dominated by the vertices associated with the common bits of $\neg \ell$ and $i$, which is unique since each number has a unique binary representation; moreover $x_\ell$ cannot be confounded with vertices within the bit gadgets since $z \in N_G(x_\ell) \cap S$.
    As to the vertices of $Y$, $z'$ is the only vertex dominated only by $z$, since $y_1$ and $y_2$ are also dominated by $x_i$; moreover, $y_0$ is dominated by the same vertices as $x_{\neg i}$ but it holds that $N_G(y_0) \setminus N_G(x_{\neg i}) = \{x_i\}$, so $y_0$ cannot be confounded with another vertex of $G$.
    Finally, $y_2$ is the only vertex of $G$ that has only $x_i$ in its neighborhood: all others dominated by $x_i$ are either in $T$, which have at least one more vertex non-adjacent to $y_2$, or are in an hyperedge gadget, where the same property holds.
    By Claim~\ref{clm:fw_sub}, the vertices of $G^*$ are satisfied only by vertices of $S \cap V(G^*)$, so increasing $S$ cannot create any confusion among them, which concludes the proof.
\end{proof}

\begin{observation}
    \label{obs:three_b_choice}
    For every solution $S$ of $(G, d)$ and binary choice gadget $G_\rho(\alpha, \beta)$ whose only vertices with neighbors outside of $G_\rho(\alpha, \beta)$ are $\alpha_\rho$ and $\beta_\rho$, it holds that $S \cap \{a_\rho, \alpha_\rho, \beta_\rho\} \neq \emptyset$, $S \cap \{b_\rho, c_\rho\} \neq \emptyset$, and $S \cap \{d_\rho, e_\rho\} \neq \emptyset$.
\end{observation}

\begin{observation}
    \label{obs:one_hyperedge}
    For every solution $S$ and hyperedge gadget $G_\varepsilon$ whose only vertices with neighbors outside of $G_\varepsilon$ are $A_\varepsilon$ and $B_\varepsilon$, it holds that $S \cap \{c_\rho, r_\rho\} \neq \emptyset$.
\end{observation}

\begin{lemma}
    \label{lem:t_acc_choice}
    In any solution $S$ to $(G,d)$, $|S \cap V(T)| \geq 3h + 2$.
    Furthermore, if $|S \cap V(T)| = 3h + 2$ and $N_S(X) \setminus V(T) = \emptyset$, then $z \in S$, there exists an unique $x_i \in S$, and the picked vertices of $\mathcal{T}$ correspond to the binary representation of $i$.
\end{lemma}

\begin{proof}
    For the first statement, suppose there is some solution $S$ with $|S \cap V(T)| = 3h + 1$.
    By Observation~\ref{obs:three_b_choice}, we know that $|S \cap V(\mathcal{T})| \geq 3h$, and cannot be larger than $3h$ otherwise $z'$ is not dominated.
    Similarly, if $x_i \in S$, $z'$ is not dominated, so we conclude that $X \cap S = \emptyset$.
    Moreover, picking only one vertex of $Y$ is not enough to cover all of $Y$: either $y_2$ or $z'$ are not dominated, so $S$ is not a solution.
    
    For the second part of the statement, first suppose that $z \notin S$, from which we conclude that $z' \in S$, otherwise $z'$ is not dominated, so we only have one more vertex to consider.
    Now, if $\{y_0, y_1\} \cap S \neq \emptyset$, $y_2$ is not dominated but, if $y_2 \in S$, $y_1$ is not dominated, so there must be some $x_i \in X \cap S$, but in this case, $y_1$ and $y_2$ are dominated only by $x_i$, a contradiction; thus, we may now assume that $z \in S$.
    Note that we cannot have more than one vertex in $S \cap X$, otherwise $|S \cap V(T)| > 3h +1$, so suppose that $X \cap S = \emptyset$, which actually implies that $y_2 \in S$, otherwise $y_2$ would not be dominated; however, $y_1$ and $z'$ now satisfy $N_S(y_1) = N_S(z') = \{z\}$, another contradiction.
    
    Assume now that $x_i \in S$ but there is some $T_j(0,1)$ where $\{0_j, 1_j\} \cap S = \emptyset$ and suppose w.l.o.g. that $0_j$ is adjacent to $x_i$.
    Since, by Observation~\ref{obs:three_b_choice} and our previous argumentation, $|S \cap V(\mathcal{T})| = 3h$ and $|\{0_j, 1_j\} \cap S| \leq 1$, there is necessarily some $x_\ell$ where $N_S(x_\ell) \cap V(\mathcal{T}) = \emptyset$, i.e $N_S(x_\ell) = \{z\} = N_S(z')$ since $N_S(x_\ell) \subseteq V(T)$, a contradiction; this implies that $|\{0_j, 1_j\} \cap S| = 1$ for every $j$.
    Suppose that there is non-empty set of indices $J  \subseteq \{0\} \cup [h-1]$ where, for every $j \in J$, $\gamma_j \in N_S(x_i)$, i.e the vertices of $S \cap V(\mathcal{T})$ do not correspond to the bits of $x_i$.
    Since $|\mathcal{H}|$ is a power of two, there is some $x_\ell$ that, for every $j \notin J$, has $\gamma_j \in N_G(x_\ell) \cap N_G(x_i)$, i.e $\gamma_j \notin S$, and for every $j \in J$, has $\gamma_j \notin N_G(x_\ell)$, which in turn implies that $N_S(x_\ell) = \{z\} = N_S(z')$, again because $N_S(x_\ell) \subseteq V(T)$, which is a contradiction that completes the proof.
\end{proof}

\begin{lemma}
    \label{lem:conv_no_kernel}
    If $(G,d)$ admits a solution, then at least one $H_i$ admits a valid bicoloring.
\end{lemma}

\begin{proof}
    Let $S$ be a solution to $(G,d)$.
    We claim that, for every $\varepsilon \in \mathcal{E}$, $\{A_\varepsilon, B_\varepsilon\} \cap S = \emptyset$.
    By Observation~\ref{obs:three_b_choice} and by the first part of Lemma~\ref{lem:t_acc_choice}, we need at least $3(n + h) + 2$ vertices to cover all vertices outside of hyperedge gadgets and, by Observation~\ref{obs:one_hyperedge}, we need one vertex of each of the $m$ hyperedge gadgets different from $A_\varepsilon$ and $B_\varepsilon$, so if there is some hyperedge gadget with more than one vertex in $S$, it must be the case that $S \geq 3(n + h) + 2 + m + 1 > d$, a contradiction.
    We also claim that, for every $v \in U$, $|\{\alpha_v, \beta_v\} \cap S| = 1$.
    To see that this is the case, again using Observation~\ref{obs:three_b_choice} and the previously shown fact that neither $A_\varepsilon$ nor $B_\varepsilon$ are in $S$, for any $\varepsilon \in \mathcal{E}$, if there is some $v$ where $\{\alpha_v, \beta_v\} \cap S = \emptyset$, then $a_v \in S$, but in this case $N_S(\alpha_v) = N_S(\beta_v)$, a contradiction.
    As such, we define a bicoloring $\varphi$ where $v \in \varphi_\alpha$ if and only if $\alpha_v \in S$.
    Since $N_S(X) \setminus V(T) = \emptyset$ and $|S \cap V(T)| = 3h + 2$, there is an unique $x_i \in X \cap S$.
    We claim that $\varphi$ is a solution to $H_i$.
    Suppose that $\varepsilon \in E(H_i)$ is a monochromatic hyperedge with $\varepsilon = \{u,v,w\}$ and $\varphi(v) = \alpha$.
    If $c_\varepsilon \in S$, we know that $r_\varepsilon \notin S$, otherwise $S > d$, and, since $x_iB_\varepsilon \notin E(G)$, it holds that $N_S(B_\varepsilon) = N_S(r_\varepsilon) = c_\varepsilon$, a contradiction; otherwise, if $c_\varepsilon \notin S$, $r_\varepsilon \in S$, but then $N_S(c_\varepsilon) = \emptyset$, another contradiction, so it must be the case that $\varepsilon$ does not exist, so $\varphi$ is a proper bicoloring of $H_i$.
\end{proof}

\begin{lemma}
    \label{lem:bound_vc}
    $G$ has a vertex cover with $\bigO{n^3 + \log |\mathcal{H}|}$ vertices.
\end{lemma}

\begin{proof}
    Since each $H_i$ is 3-uniform, there are $\bigO{n^3}$ hyperedge gadgets.
    Summing the $n$ vertex gadgets, the other $\log_2 |\mathcal{H}|$ binary selection gadgets of $\mathcal{T}$, and the $\bigO{1}$ auxiliary vertices $Y$, we have a set $K$ of size $\bigO{n^3 + \log |\mathcal{H}|}$.
    Since $V(G) \setminus K = X$ and $X$ is an independent set, $G$ has a vertex cover of size at most $|K|$.
\end{proof}

Bringing together Lemmas~\ref{lem:fw_no_kernel},~\ref{lem:conv_no_kernel}, and~\ref{lem:bound_vc}, Theorem~\ref{thm:no_kernel_vc} follows immediately.

\begin{theorem}
    \label{thm:no_kernel_vc}
    \textsc{Locating-Dominating Set} does not admit a polynomial kernel when parameterized by vertex cover and maximum size of the solution unless NP $\subseteq$ coNP/poly.
\end{theorem}

\subsection{Distance to Clique}

We now show how to replace the gadget of Figure~\ref{fig:sel_gadget} for a gadget that allows us to conclude that \textsc{Locating-Dominating Set} has no polynomial kernel when parameterized by distance to clique unless NP $\subseteq$ coNP/poly.
Let $T$ be our instance selector gadget.
Once again, we begin by adding $h = \log_2 |\mathcal{H}|$ copies $\mathcal{T} = \{T_0(0,1), \dots, T_{h-1}(0,1)\}$ of the gadget in Figure~\ref{fig:vertex_gadget} to $T$, then add $|\mathcal{H}|$ vertices $X = \{x_0, \dots, x_t\}$, and two auxiliary vertices $Y = \{y_0, z\}$, which are adjacent to each other, to the gadget.
We connect vertices in $X$ to vertices of $\mathcal{T}$ as before, i.e. if the $j$-th bit of $i$ is equal to zero, add the edge $x_i1_j$, otherwise add edge $x_i0_j$.
Now, add all edges between $z \in Y$ and $X$, and, for each $T_j \in \mathcal{T}$, add the edges $y_00_j$ and $y_01_j$.
For each hyperedge $\varepsilon \in \mathcal{E}$ and $H_i \in \mathcal{H}$, if $\varepsilon \notin E_i$, add the edges $x_iA_\varepsilon$ and $x_iB_\varepsilon$.
Finally, set $d = 3(n + h) + m + 1$.
We present an example for $h=2$ in Figure~\ref{fig:sel_gadget_clique}.

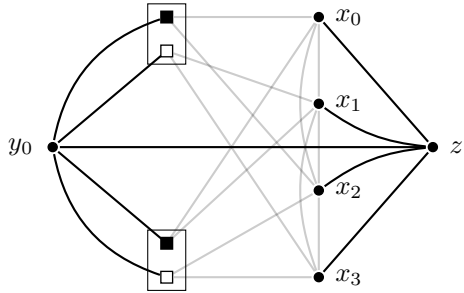
\begin{figure}[!htb]
    \centering
        \begin{tikzpicture}[scale=1]
            %\draw[help lines] (-5,-5) grid (5,5);
            \GraphInit[unit=3,vstyle=Normal]
            \SetVertexNormal[Shape=circle, FillColor=black, MinSize=3pt]
            \tikzset{VertexStyle/.append style = {inner sep = \inners, outer sep = \outers}}
            \SetVertexLabelOut
            \Vertex[x=1, y=1.725, Lpos=0,Math]{x_0}
            \Vertex[x=1, y=0.575, Lpos=0,Math]{x_1}
            \Vertex[x=1, y=-0.575, Lpos=0,Math]{x_2}
            \Vertex[x=1, y=-1.725, Lpos=0,Math]{x_3}
            
            \Vertex[x=-2.5cm, y=0, Lpos=180, Ldist=1pt,Math, L={y_0}]{y}
            
            \Vertex[x=2.5, y=0, Lpos=0, Math]{z}
            
            \Edges(x_3, z, x_0)
            \Edge(z)(y)
            \Edges[style={bend left=15}](x_2, z, x_1)
            \Edges[style={opacity=0.2}](x_0, x_1, x_2, x_3)
            \Edges[style={opacity=0.2, bend right=20}](x_0, x_2)
            \Edges[style={opacity=0.2, bend right=20}](x_1, x_3)

            \SetVertexNoLabel
            \begin{scope}[xshift=-1cm, rotate=90]
                \Vertex[x=-1.725, y=0, Lpos=0,Math]{al_0}
                \Vertex[x=-1.275, y=0, Lpos=0,Math]{be_0}
                \draw (-1.10, -0.25) rectangle (-1.9, 0.25);
                
                \Vertex[x=1.275, y=0, Lpos=0,Math]{al_1}
                \Vertex[x=1.725, y=0, Lpos=0,Math]{be_1}
                \draw (1.10, -0.25) rectangle (1.9, 0.25);
            \end{scope}
            \Edges(be_0, y)
            \Edge[style={bend right}](y)(al_0)
            \Edge[style={bend left}](y)(be_1)
            \Edge(y)(al_1)
            \Edges[style={opacity=0.2}](be_0, x_0, be_1)
            \Edges[style={opacity=0.2}](be_0, x_1, al_1)
            \Edges[style={opacity=0.2}](al_0, x_2, be_1)
            \Edges[style={opacity=0.2}](al_0, x_3, al_1)
            \begin{scope}
                \tikzset{VertexStyle/.append style = {shape = rectangle, inner sep = 2.2pt}}
                \AddVertexColor{black}{be_0, be_1}
                \AddVertexColor{white}{al_0, al_1}
            \end{scope}
            
        \end{tikzpicture}
\caption{Instance selector gadget for $h = 2$. Black rectangle are vertices of the form $1_j$, while white rectangles the $0_j$ vertices.\label{fig:sel_gadget_clique}}
\end{figure}

If we prove an analogue to Lemma~\ref{lem:t_acc_choice} we are done, since: (i) the remainder of the argument will be the same as in Lemmas~\ref{lem:fw_no_kernel} and~\ref{lem:conv_no_kernel}, (ii) $X$ is a clique, and (iii) $|V(G) \setminus X|$ and $d$ are polynomials on $(n + \log_2 |\mathcal{H}|)$.
We do so in the following lemma.

\begin{lemma}
    \label{lem:t_acc_choice_clique}
    In any solution $S$ to $(G,d)$, $|S \cap V(T)| \geq 3h + 1$.
    Furthermore, if $|S \cap V(T)| = 3h + 1$ and $N_S(X) \setminus V(T) = \emptyset$, then there exists an unique $x_i \in S$, and the picked vertices of $\mathcal{T}$ correspond to the binary representation of $i$.
\end{lemma}

\begin{proof}
    For the first statement, suppose there is some solution $S$ with $|S \cap V(T)| = 3h$.
    By Observation~\ref{obs:three_b_choice}, we know that $|S \cap V(\mathcal{T})| \geq 3h$, which implies $S \cap V(T) = S \cap V(\mathcal{T})$, but in this case $z$ is not dominated by any vertex.
    
    For the second statement, suppose that $X \cap S = \emptyset$.
    If $S \cap V(\mathcal{T}) = 3h + 1$, then $z$ is not dominated.
    Consequently, at most one vertex of $\{0_j, 1_j\}$ was picked for each $T_j \in \mathcal{T}$ and, since $|\mathcal{H}|$ is a power of two, there is some $x_i \in X$ that has no neighbor in $S \cap V(\mathcal{T})$ and some $x_\ell \in X$ where $N_S(x_\ell) \cap V(\mathcal{T}) = N_S(y_0) \cap V(\mathcal{T})$; the former implies that $z \in S$, otherwise $x_i$ is not dominated, while the latter implies that $y_0 \in S$, otherwise $x_\ell$ is confounded with $y_0$, which contradicts the hypothesis that $|S| = 3h + 1$, so it must hold that $X \cap S \neq \emptyset$.
    
    Suppose now that there is some $T_j \in \mathcal{T}$ where $\{0_j, 1_j\} \cap S = \emptyset$ and note that there are two $x_i, x_\ell \in X$ where $N_T[x_i] \setminus \{0_j\} = N_T[x_i] \setminus \{1_j\}$ and, since $N_S(X) \setminus V(T) = \emptyset$, if neither $x_i$ nor $x_\ell$ are in $S$, then $N_S(x_i) = N_S(x_\ell)$; if, on other hand, $x_i \in S$, $N_S(z) = N_S(x_\ell)$.
    Either way, we have a contradiction to the hypothesis that $S$ is a solution to $(G,d)$, so we must have that $|\{0_j, 1_j\} \cap S| = 1$.
    Consequently, let $x_i \in X \cap S$ and note that, if $x_i$ has a neighbor in $S$, then it is one of $\{0_j, 1_j\}$ for some $T_j \in \mathcal{T}$, and there is another $x_\ell \in X \setminus S$ where $N_S(x_\ell) = \{x_i\}$, but we also have that $N_S(z) = \{x_i\}$, a contradiction.
    So $x_i \in X \cap S$ has no neighbors in $S \cap V(\mathcal{T})$, which implies that, for every $T_j \in \mathcal{T}$, if the $j$-th bit of $i$ is $\gamma$, then $\gamma_j \in V(T_j) \cap S$, so the picked vertices of $\mathcal{T}$ correspond to the binary representation of $i$.
\end{proof}

\begin{theorem}
    \label{thm:no_kernel_dc}
    \textsc{Locating-Dominating Set} does not admit a polynomial kernel when parameterized by distance to clique and maximum size of the solution unless NP $\subseteq$ coNP/poly.
\end{theorem}
\section{A linear kernel for max leaf number}

We now present a linear kernel for \textsc{Locating-Dominating Set} for the max leaf number parameter, denoted by $\ml(\cdot)$.
We make heavy use of the fact that, if $G$ has max leaf number $k$, then $G$ is the subdivision of a graph $H$ on $4k$ vertices~\cite{max_subdivision}; $H$ is called the \textit{host graph} of $G$. 
Since computing such a graph is an \NPH\  problem~\cite{garey_johnson}, we assume that $H$ is part of the input of our kernelization algorithm.

Throughout this section, let $H$ be the graph that $G$ is a subdivision of, $d$ the maximum size of a locating-dominating set, $P(u,v)$ be the path that replaced $uv \in E(H)$ in order to obtain $G$ --- note that $u,v \notin P(u,v)$ --- and let $\mathcal{P}_2(G)$ be the connected components of $G - H$, i.e. the paths formed by vertices of $G - H$.
Before proceeding to the kernelization algorithm, we show a bound on the size of $\mathcal{P}_2(G)$ which was not known to the best of our knowledge.

\begin{lemma}
    \label{lem:ml_bound}
    If $G$ has max leaf number $k$ and $H$ is its host graph, then $|\mathcal{P}_2(G)| \leq 5k - 1 + \floor{\frac{k}{2}}$.
\end{lemma}

\begin{proof}
    Let $T(G)$ be a spanning tree of $G$ with  $k$ leaves and $L$ be its leaf set.
    Note that, if exactly one leaf $a \in V(G)$ is contained in the path $P(u,v) \in \mathcal{P}_2(G)$, then we may assume that $a$ is adjacent to either $u$ or $v$, otherwise we would necessarily have a second leaf $b \in P(u,v)$.
    As such, we define graph $H'$ to be the graph where $V(H') = V(H) \cup L$ and its edge set is given as follows.
    If $uv \in E(H) \cap E(G)$ or $P(u,v) \cap L = \emptyset$, then $uv$ is in $H'$; otherwise, there is some $a \in P(u,v) \cap L$: if $a$ is the unique leaf in $P(u,v)$, then $H'$ contains both $ua$ and $av$, if $b$ is also a leaf in $P(u,v)$, closer to $v$ than $a$, then $H'$ contains $ua$, $ab$ and $bv$.
    Note that $\ml(H') = \ml(G)$, since $H'$ admits a spanning tree on the same leaf set as $G$.
    Now, let $w : E(H') \mapsto \mathbb{Z}$ be a weight function where, for every $uv \in E(H')$, $w(uv) = |P(u,v)|$, if $P(u,v) \in \mathcal{P}_2(G)$, or $w(uv) = 0$, otherwise.
    Also, let $T(H')$ be a spanning tree of $H'$ that has $L$ as its leaf set and minimizes $\max_{uv \in E(T(H'))}\{w(uv)\}$, and $T(G)$ the subtree of $G$ that contains $V(H')$ and replaces edges $uv \in T(H')$ with the corresponding paths of $G$. 
    Suppose that there is some edge $uv \in E(H') \setminus E(T(H'))$ with $w(uv) > 0$, and let $u', v'$ be the neighbors of $u$ and $v$ in $P(u,v)$, respectively.
    If $u$ is not a leaf of $T(H')$, we claim that $G$ has $\ml(G) \geq k + 1$: since $w(uv) > 0$, we can find another subtree $T'(G)$ of $G$ that includes all vertices and edges of $T(G)$, all vertices of $P(u,v)$, and all edges of $P(u,v)$ except $v'v$, so $T'(G)$ has all $k$ leaves of $T(H')$ and the additional leaf $v'$;  this contradicts the hypothesis that $\ml(G) = k$ since it is always possible to extend a subtree of $G$ on $\ell$ leaves to a spanning tree of $G$ on $\ell$ leaves~\cite{bonsma_max_leaf}.
    Now we are left with the case where both $u,v \in L$; our claim is that, if $w(uv) > 0$, then there is no other edge $uz$ with $w(uz) > 0$.
    Suppose to the contrary: let $u''$ be the neighbor of $u$ in $P(u,z)$ and $z''$ the neighbor of $z$ in $P(u,z)$.
    Then there is a subtree $T'(G)$ that has all vertices and edges of $T(G)$, all vertices of $P(u,v) \cup P(u,z)$, and all edges of $P(u,v) \cup P(u,z)$ except for $z''z$ and $v'v$, i.e. $u$ is not a leaf of $T'(G)$ but both $v'$ and $z''$ are, so $T'(G)$ has $k+1$ leaves, contradicting the hypothesis that $G$ has $\ml(G) = k$.
    
    By our analysis, each leaf of $T(H')$ may have at most one edge in $e \in E(H') \setminus E(T(H'))$ where $w(e) > 0$ and no internal vertex of $T(H')$ may be incident to such an edge.
    Since edges have two endpoints, $T(H')$ has $k$ leaves and at most $5k$ vertices, we have that $|\{e \in V(H) \mid w(e) > 0\}| \leq 5k - 1 + \floor{\frac{k}{2}}$.
\end{proof}

For each $P(u,v) \in \mathcal{P}_2(G)$, a \textit{5-section} of $P$ is a subpath of $P$ of length five.
A \textit{5-sectioning} of $P$ is a partition of $P$ into $t$ sets $\varphi(P) = \{\varphi_1, \dots, \varphi_\tp\}$ where: (i) $t$ is maximum, (ii) the vertices of each $\varphi_i$ are denoted by $\{x_j^i \mid j \in \left[t(P)\right]\}$, (iii) for every $i \in [t(P)-1] \setminus \{1\}$, each $\varphi_i$ is a 5-section and $x_5^i$ is adjacent to $x_1^{i+1}$, and (iv) both $\varphi_1$ and $\varphi_\tp$ have at least five vertices and $||\varphi_1| - |\varphi_\tp||$ is minimum.
We denote by $\varphi^*(P)$ the inner 5-sections of $P$, i.e. $\varphi^*(P) = \varphi(P) \setminus \{\varphi_1, \varphi_\tp\}$.
In an abuse of notation, we define $V(\varphi^*(P)) = \{x \in \varphi_i \mid \varphi_i \in \varphi^*(P)\}$.
Note that a path $P$ on at least 15 vertices has at most four different 5-sectionings; our arguments are indifferent to which of them was.
The algorithm works by identifying long paths obtained through the subdivision of $H$ and replacing them with shorter paths whose solutions are easily extended. 
In particular, the \textit{reverse} of a 5-sectioning $\varphi(P)$, denoted by $\rev(\varphi(P)) = \phi(P)$, satisfies $\phi_i = \varphi_{\tp - i + 1}$.
Our first goal is to normalize solutions in such a way that extending them becomes a trivial task.

\begin{lemma}
    \label{lem:normalize_ml}
    Let $(G,d)$ be an instance of \pname{Locating-Dominating Set}, $H$ be the graph that $G$ is a subdivision of, $P \in \mathcal{P}_2(G)$, and $S$ a solution to $(G,d)$.
    If $|P| \geq 15$, then there is a solution $S'$ with $|S'| \leq |S|$ where, for each $\varphi_i \in \varphi^*(P)$, $|\varphi_i \cap S'| \leq 2$ and, if $x_j^2, x_\ell^2 \in S'$, then for every $\varphi_i \in \varphi^*(P)$, $\varphi_i \cap S' = \{x_j^i, x_\ell^i\}$.
\end{lemma}

\begin{proof}
    We begin by setting $S' = S \setminus (P \setminus \varphi_1 \setminus \varphi_{t(p)})$.
    To simplify our notation, we assume that $|\varphi_1| = 5$; as we shall see, our case analysis needs only to consider the four vertices of $\varphi_1$ closer to $x_1^2$.
    We branch on the intersection between $S$, $\varphi_1$ and $\varphi_\tp$.  
    \begin{enumerate}
        \item{\label{case:normalize_ml1}} If $x_5^1 \in S'$ and $\{x_1^\tp, x_2^\tp\} \cap S' \neq \emptyset$, add $x_2^i$ and $x_5^i$ to $S'$ for every $\varphi_i \in \varphi^*(P)$.
        To see that this does not confound any pair of vertices belonging to the 5-sections in $\varphi^*(P)$, we refer to Figure~\ref{fig:normalize_ml1}.
        Since one of $x_1^{t(P)}, x_2^{t(P)}$ is in $S'$, $N_{S'}(x_4^{t(P)-1})$ and $N_{S'}(x_1^\tp)$ are distinct, $N_{S'}(x_2^\tp) = N_S(x_2^\tp)$ is either unique or $x_2^\tp \in S'$, and all other vertices of $x \in \varphi_\tp$ cannot be confounded since $N_{S'}(x) = N_S(x)$.
        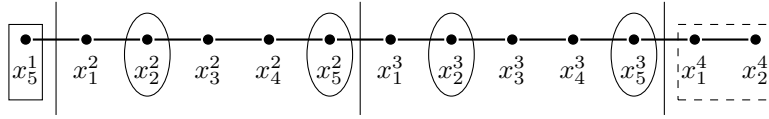
\begin{figure}[!htb]
            \centering
                \begin{tikzpicture}[scale=1]
                    %\draw[help lines] (-5,-5) grid (5,5);
                    \GraphInit[unit=3,vstyle=Normal]
                    \SetVertexNormal[Shape=circle, FillColor=black, MinSize=3pt]
                    \tikzset{VertexStyle/.append style = {inner sep = \inners, outer sep = \outers}}
                    \SetVertexLabelOut
                    %\draw (0, 1) -- (0, -1);
                    
                    \foreach \j in {2,3} {
                        \foreach \i in {1,2,3,4,5} {
                            \pgfmathsetmacro{\x}{(5*(\j-1) + \i)*0.8}
                           
                            \ifthenelse{\i = 2 \OR \i = 5} {
                                \begin{scope}[xshift=\x cm]
                                \draw (0,-0.2) ellipse (0.3cm and 0.55cm);
                                \end{scope}
                            }{}
                            \Vertex[x=\x, y=0, Lpos=270,Math, L={x_\i^\j}]{x_\i\j}
                        }
                    }
                    \draw (4.4, -1) -- (4.4,0.5);
                    \draw (8.4, -1) -- (8.4,0.5);
                    \draw (12.4, -1) -- (12.4,0.5);
                    \Vertex[x=4, y=0, Lpos=270,Math,L={x_5^1}]{x_51}
                    \draw (3.78, -0.8) rectangle (4.22, 0.2);
                    
                    \Vertex[x=12.8, y=0, Lpos=270,Math,L={x_1^4}]{x_14}
                    \Vertex[x=13.6, y=0, Lpos=270,Math,L={x_2^4}]{x_24}
                    
                    \draw[dashed] (12.58, -0.8) rectangle (13.82, 0.2);
                    
                    \Edges(x_51, x_12, x_22, x_32, x_42, x_52, x_13, x_23, x_33, x_43, x_53, x_14, x_24)
                    %\Edge(x_11)(x_45)
                \end{tikzpicture}
        \caption{Case~\ref{case:normalize_ml1} of the proof of Lemma~\ref{lem:normalize_ml}.
        Solid boxes indicate that the boxed vertex is in $S \cap S'$, while solid ellipses indicate vertices in $S'$ but not necessarily in $S$.
        The large dashed box indicates that at least one of the boxed vertices is in the solution.\label{fig:normalize_ml1}}
        \end{figure}
        
        \item{\label{case:normalize_ml2}} If $x_5^1 \in S$ but $\{x_1^\tp, x_2^\tp\} \cap S' = \emptyset$, i.e. $x_3^\tp \in S$, we have two subcases.
        \begin{enumerate}
            \item{\label{case:normalize_ml2-1}} If $N_{S}[x_4^1] \setminus \{x_5^1\} \neq \emptyset$, then add $x_3^i$ and $x_5^i$ for every $\varphi_i \in \varphi^*(P)$.
            Vertex $x_1^2$ cannot be confounded with vertex $x_4^1$ since either $x_4^1 \in S'$ or $x_3^1 \in S$, which implies that $x_1^2$ is the unique vertex with $N_{S'}(x_1^2) = \{x_5^1\}$.
            For the remaining vertices in $V(\varphi^*(P))$, we prove by example using Figure~\ref{fig:normalize_ml2a}.
            Note that $|N_{S'}(x_4^{t(P) - 1})| = 2$, so $x_1^\tp$ is the unique vertex whose only neighbor in $S'$ is $x_5^{t(P)-1}$.
            Finally, for the vertices $x \in \varphi_\tp \setminus \{x_1^\tp\}$, $N_{S'}(x) = N_S(x)$, so they cannot be confounded.
            \begin{figure}[!htb]
                \centering
                    \begin{tikzpicture}[scale=1]
                        %\draw[help lines] (-5,-5) grid (5,5);
                        \GraphInit[unit=3,vstyle=Normal]
                        \SetVertexNormal[Shape=circle, FillColor=black, MinSize=3pt]
                        \tikzset{VertexStyle/.append style = {inner sep = \inners, outer sep = \outers}}
                        \SetVertexLabelOut
                        %\draw (0, 1) -- (0, -1);
                        
                        \foreach \j in {2,3} {
                            \foreach \i in {1,2,3,4,5} {
                                \pgfmathsetmacro{\x}{(5*(\j-1) + \i)*0.8}
                               
                                \ifthenelse{\i = 3 \OR \i = 5} {
                                    \begin{scope}[xshift=\x cm]
                                        \draw (0,-0.2) ellipse (0.3cm and 0.55cm);
                                    \end{scope}
                                }{}
                                \Vertex[x=\x, y=0, Lpos=270,Math, L={x_\i^\j}]{x_\i\j}
                            }
                        }
                        \foreach \j in {1} {
                            \foreach \i in {3,4,5} {
                                \pgfmathsetmacro{\x}{(5*(\j-1) + \i)*0.8}
                               
                                \ifthenelse{\i = 5} {
                                    \begin{scope}[xshift=\x cm]
                                        \draw (-0.22, -0.8) rectangle (0.22, 0.2);
                                    \end{scope}
                                }{}
                                \Vertex[x=\x, y=0, Lpos=270,Math, L={x_\i^\j}]{x_\i\j}
                            }
                        }
                        \foreach \j in {4} {
                            \foreach \i in {1,2,3} {
                                \pgfmathsetmacro{\x}{(5*(\j-1) + \i)*0.8}
                               
                                \ifthenelse{\i = 3} {
                                    \begin{scope}[xshift=\x cm]
                                        \draw (-0.22, -0.8) rectangle (0.22, 0.2);
                                    \end{scope}
                                }{}
                                \Vertex[x=\x, y=0, Lpos=270,Math, L={x_\i^\j}]{x_\i\j}
                            }
                        }
                        \draw (4.4, -1) -- (4.4,0.5);
                        \draw (8.4, -1) -- (8.4,0.5);
                        \draw (12.4, -1) -- (12.4,0.5);
                        
                        \draw[dashed] (2.18, -0.8) rectangle (3.42, 0.2);

                        \Edges(x_31, x_41, x_51, x_12, x_22, x_32, x_42, x_52, x_13, x_23, x_33, x_43, x_53, x_14, x_24, x_34)
                        %\Edge(x_11)(x_45)
                    \end{tikzpicture}
            \caption{Case~\ref{case:normalize_ml2-1} of the proof of Lemma~\ref{lem:normalize_ml}.
            Solid boxes indicate that the boxed vertex is in $S \cap S'$, while solid ellipses indicate vertices in $S'$ but not necessarily in $S$.
            The large dashed box indicates that at least one of the boxed vertices is in the solution.\label{fig:normalize_ml2a}}
            \end{figure}
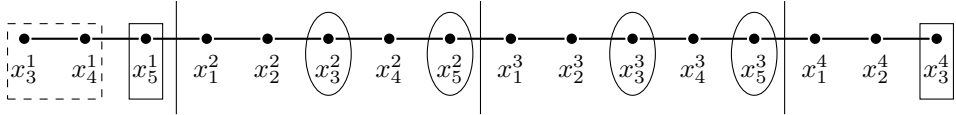
            
            \item{\label{case:normalize_ml2-2}} If $N_{S'}[x_4^1] = \{x_5^1\}$, we first claim that there is some $\varphi_\ell \in \varphi^*(P)$ with $|\varphi_\ell \cap S| \geq 3$.
            Note that $x_5^{\tp - 1} \in S$, otherwise $x_1^\tp$ is not dominated.
            Let $r$ be the largest index such that $\varphi_r \cap S \neq \{x_3^r, x_5^r\}$.
            If $r = \tp - 1$, then $|\varphi_r \cap S| \geq 3$ since at least one of $\{x_1^{\tp - 1}, x_2^{\tp - 1}, x_3^{\tp - 1}\}$ is required to dominate $x_2^{\tp - 1}$.
            If $r < \tp - 1$, we must have $x_5^{r} \in S$, otherwise $x_1^{r+1}$ is not dominated; moreover, if $x_3^r \in S$, then a third $x_j \in \varphi_r$ is in $S$, since we are assuming $\varphi_r \cap S \neq \{x_3^r, x_5^r\}$.
            However, if $x_3^r \notin S$, then $x_4^r \in S$, otherwise $x_1^{r+1}$ would be confounded with $x_4^r$ in $S$, but in this case at least one of $\{x_1^r, x_2^r, x_3^r\}$ is also required to dominated $x_2^r$, implying $|\varphi_r \cap S| \geq 3$.
            As such, we add $x_3^1$ to $S'$ and proceed as in Case~\ref{case:normalize_ml2-1}.
            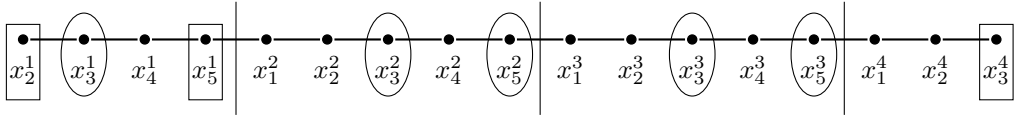
\begin{figure}[!htb]
                \hspace{-1cm}
                    \begin{tikzpicture}[scale=1]
                        %\draw[help lines] (-5,-5) grid (5,5);
                        \GraphInit[unit=3,vstyle=Normal]
                        \SetVertexNormal[Shape=circle, FillColor=black, MinSize=3pt]
                        \tikzset{VertexStyle/.append style = {inner sep = \inners, outer sep = \outers}}
                        \SetVertexLabelOut
                        %\draw (0, 1) -- (0, -1);
                        
                        \foreach \j in {1} {
                            \foreach \i in {2,3,4,5} {
                                \pgfmathsetmacro{\x}{(5*(\j-1) + \i)*0.8}
                               
                                \ifthenelse{\i = 3} {
                                    \begin{scope}[xshift=\x cm]
                                        \draw (0,-0.2) ellipse (0.3cm and 0.55cm);
                                    \end{scope}
                                }{}
                                \ifthenelse{\i = 2 \OR \i = 5} {
                                    \begin{scope}[xshift=\x cm]
                                        \draw (-0.22, -0.8) rectangle (0.22, 0.2);
                                    \end{scope}
                                }{}
                                \Vertex[x=\x, y=0, Lpos=270,Math, L={x_\i^\j}]{x_\i\j}
                            }
                        }
                        \foreach \j in {2,3} {
                            \foreach \i in {1,2,3,4,5} {
                                \pgfmathsetmacro{\x}{(5*(\j-1) + \i)*0.8}
                               
                                \ifthenelse{\i = 3 \OR \i = 5} {
                                    \begin{scope}[xshift=\x cm]
                                        \draw (0,-0.2) ellipse (0.3cm and 0.55cm);
                                    \end{scope}
                                }{}
                                \Vertex[x=\x, y=0, Lpos=270,Math, L={x_\i^\j}]{x_\i\j}
                            }
                        }
                        \foreach \j in {4} {
                            \foreach \i in {1,2,3} {
                                \pgfmathsetmacro{\x}{(5*(\j-1) + \i)*0.8}
                               
                                \ifthenelse{\i = 3} {
                                    \begin{scope}[xshift=\x cm]
                                        \draw (-0.22, -0.8) rectangle (0.22, 0.2);
                                    \end{scope}
                                }{}
                                \Vertex[x=\x, y=0, Lpos=270,Math, L={x_\i^\j}]{x_\i\j}
                            }
                        }
                        \draw (4.4, -1) -- (4.4,0.5);
                        \draw (8.4, -1) -- (8.4,0.5);
                        \draw (12.4, -1) -- (12.4,0.5);
                        
                        %\draw[dashed] (2.18, -0.8) rectangle (3.42, 0.2);

                        \Edges(x_21, x_31, x_41, x_51, x_12, x_22, x_32, x_42, x_52, x_13, x_23, x_33, x_43, x_53, x_14, x_24, x_34)
                        %\Edge(x_11)(x_45)
                    \end{tikzpicture}
            \caption{Case~\ref{case:normalize_ml2-2} of the proof of Lemma~\ref{lem:normalize_ml}.
            Solid boxes indicate that the boxed vertex is in $S \cap S'$, while solid ellipses indicate vertices in $S'$ but not necessarily in $S$.\label{fig:normalize_ml2-2}}
            \end{figure}
        \end{enumerate}
        
        \item{\label{case:normalize_ml3}} If $x_5^1 \notin S$, $x_4^1 \in S$, $N_{S}[x_3^1] \setminus \{x_4^1\} \neq \emptyset$, and $x_1^\tp \in S$, then $\rev(\varphi)$ satisfies the condition of Case~\ref{case:normalize_ml1}, i.e. we complete $S'$  by adding $x_1^i$ and $x_4^i$ to it, as shown in Figure~\ref{fig:normalize_ml3}. 
        \begin{figure}[!htb]
            \centering
                \begin{tikzpicture}[scale=1]
                    %\draw[help lines] (-5,-5) grid (5,5);
                    \GraphInit[unit=3,vstyle=Normal]
                    \SetVertexNormal[Shape=circle, FillColor=black, MinSize=3pt]
                    \tikzset{VertexStyle/.append style = {inner sep = \inners, outer sep = \outers}}
                    \SetVertexLabelOut
                    %\draw (0, 1) -- (0, -1);
                    
                    \foreach \j in {1} {
                        \foreach \i in {2,3,4,5} {
                            \pgfmathsetmacro{\x}{(5*(\j-1) + \i)*0.8}
                           
                            \ifthenelse{\i = 4} {
                                \begin{scope}[xshift=\x cm]
                                    \draw (-0.22, -0.8) rectangle (0.22, 0.2);
                                \end{scope}
                            }{}
                            \Vertex[x=\x, y=0, Lpos=270,Math, L={x_\i^\j}]{x_\i\j}
                        }
                    }
                    \foreach \j in {2,3} {
                        \foreach \i in {1,2,3,4,5} {
                            \pgfmathsetmacro{\x}{(5*(\j-1) + \i)*0.8}
                           
                            \ifthenelse{\i = 1 \OR \i = 4} {
                                \begin{scope}[xshift=\x cm]
                                    \draw (0,-0.2) ellipse (0.3cm and 0.55cm);
                                \end{scope}
                            }{}
                            \Vertex[x=\x, y=0, Lpos=270,Math, L={x_\i^\j}]{x_\i\j}
                        }
                    }
                    \foreach \j in {4} {
                        \foreach \i in {1} {
                            \pgfmathsetmacro{\x}{(5*(\j-1) + \i)*0.8}
                           
                            \ifthenelse{\i = 1} {
                                \begin{scope}[xshift=\x cm]
                                    \draw (-0.22, -0.8) rectangle (0.22, 0.2);
                                \end{scope}
                            }{}
                            \Vertex[x=\x, y=0, Lpos=270,Math, L={x_\i^\j}]{x_\i\j}
                        }
                    }
                    \draw (4.4, -1) -- (4.4,0.5);
                    \draw (8.4, -1) -- (8.4,0.5);
                    \draw (12.4, -1) -- (12.4,0.5);
                    
                    \draw[dashed] (1.38, -0.8) rectangle (2.62, 0.2);
                    %\draw[dashed] (12.58, -0.8) rectangle (13.82, 0.2);

                    \Edges(x_21, x_31, x_41, x_51, x_12, x_22, x_32, x_42, x_52, x_13, x_23, x_33, x_43, x_53, x_14)
                    %\Edge(x_11)(x_45)
                \end{tikzpicture}
        \caption{Case~\ref{case:normalize_ml3} of the proof of Lemma~\ref{lem:normalize_ml}.
        Solid boxes indicate that the boxed vertex is in $S \cap S'$, while solid ellipses indicate vertices in $S'$ but not necessarily in $S$.
        The large dashed box indicates that at least one of the boxed vertices is in the solution.\label{fig:normalize_ml3}}
        \end{figure}
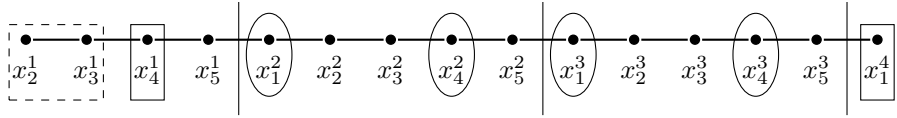
        
        \item{\label{case:normalize_ml4}} If $x_5^1 \notin S$, $x_4^1 \in S$, $N_{S}[x_3^1] \setminus \{x_4^1\} \neq \emptyset$, and $x_2^\tp \in S$, we branch in two subcases.
        \begin{enumerate}
            \item{\label{case:normalize_ml4-1}} If $N_{S}[x_3^\tp] \setminus \{x_2^\tp\} \neq \emptyset$, then proceed as in Case~\ref{case:normalize_ml3}; a very similar argument can be applied here: it suffices to note that $x_1^\tp$ is the unique vertex dominated only by $x_2^\tp$.
            \begin{figure}[!htb]
                \hspace{-1cm}
                    \begin{tikzpicture}[scale=1]
                        %\draw[help lines] (-5,-5) grid (5,5);
                        \GraphInit[unit=3,vstyle=Normal]
                        \SetVertexNormal[Shape=circle, FillColor=black, MinSize=3pt]
                        \tikzset{VertexStyle/.append style = {inner sep = \inners, outer sep = \outers}}
                        \SetVertexLabelOut
                        %\draw (0, 1) -- (0, -1);
                        
                        \foreach \j in {1} {
                            \foreach \i in {2,3,4,5} {
                                \pgfmathsetmacro{\x}{(5*(\j-1) + \i)*0.8}
                               
                                \ifthenelse{\i = 1 \OR \i = 4} {
                                    \begin{scope}[xshift=\x cm]
                                        \draw (-0.22, -0.8) rectangle (0.22, 0.2);
                                    \end{scope}
                                }{}
                                \Vertex[x=\x, y=0, Lpos=270,Math, L={x_\i^\j}]{x_\i\j}
                            }
                        }
                        \foreach \j in {2,3} {
                            \foreach \i in {1,2,3,4,5} {
                                \pgfmathsetmacro{\x}{(5*(\j-1) + \i)*0.8}
                               
                                \ifthenelse{\i = 2 \OR \i = 4} {
                                    \begin{scope}[xshift=\x cm]
                                        \draw (0,-0.2) ellipse (0.3cm and 0.55cm);
                                    \end{scope}
                                }{}
                                \Vertex[x=\x, y=0, Lpos=270,Math, L={x_\i^\j}]{x_\i\j}
                            }
                        }
                        \foreach \j in {4} {
                            \foreach \i in {1,2,3,4} {
                                \pgfmathsetmacro{\x}{(5*(\j-1) + \i)*0.8}
                                \ifthenelse{\i = 2} {
                                    \begin{scope}[xshift=\x cm]
                                        \draw (-0.22, -0.8) rectangle (0.22, 0.2);
                                    \end{scope}
                                }{}
                                \Vertex[x=\x, y=0, Lpos=270,Math, L={x_\i^\j}]{x_\i\j}
                            }
                        }
                        \draw (4.4, -1) -- (4.4,0.5);
                        \draw (8.4, -1) -- (8.4,0.5);
                        \draw (12.4, -1) -- (12.4,0.5);
                        
                        \draw[dashed] (1.38, -0.8) rectangle (2.62, 0.2);
                        \draw[dashed] (14.18, -0.8) rectangle (15.42, 0.2);

                        \Edges(x_21, x_31, x_41, x_51, x_12, x_22, x_32, x_42, x_52, x_13, x_23, x_33, x_43, x_53, x_14, x_24, x_34, x_44)
                        %\Edge(x_11)(x_45)
                    \end{tikzpicture}
            \caption{Case~\ref{case:normalize_ml4-1} of the proof of Lemma~\ref{lem:normalize_ml}.
            Solid boxes indicate that the boxed vertex is in $S \cap S'$, while solid ellipses indicate vertices in $S'$ but not necessarily in $S$.
            The large dashed box indicates that at least one of the boxed vertices is in the solution.\label{fig:normalize_ml4-1}}
            \end{figure}
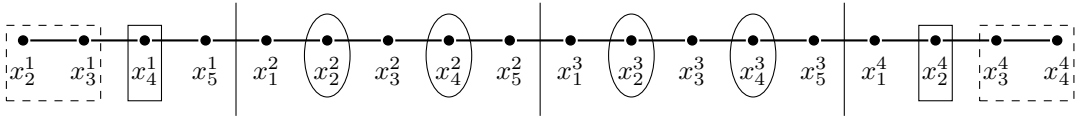
            
            \item{\label{case:normalize_ml4-2}} On the other hand, if $N_{S}[x_3^\tp] = \{x_2^\tp\}$, then we know that $x_5^{\tp - 1} \in S$ (otherwise $x_1^\tp$ is confounded with $x_3^\tp$).
            We again claim that there is some $\varphi_\ell \in \varphi^*(P)$ with at least three vertices in $S$.
            If $x_4^{\tp - 1} \in S$, then $|\varphi_{\tp-1} \cap S| \geq 3$, since at least a third vertex of $\varphi_{\tp-1}$ is required to dominate $x_2^{\tp-1}$.
            Suppose first that $x_3^{\tp-1} \in S$, let $r < \tp - 1$ be the largest index where $\varphi_r \cap S \neq \{x_3^r, x_5^r\}$, and note that: (i) $x_5^r \in S$, otherwise $x_1^{r+1}$ is not dominated, and (ii) either $x_3^r$ or $x_4^r$ must be in $S$, otherwise $x_4^r$ is confounded with $x_1^{r+1}$.
            For the former, we immediately have $|\varphi_r \cap S| \geq 3$, for the latter we conclude that at least one of $\{x_1^r, x_2^r, x_3^r\}$ is in $S$, or $x_2^r$ is not dominated, and, again, $|\varphi_r \cap S| \geq 3$.
            
            Now we may assume that $x_3^{\tp - 1} \notin S$, which implies $x_2^{\tp-1} \in S$, or $x_3^{\tp-1}$ would not be dominated.
            Let $s$ be the largest index where $\varphi_s \cap S \neq \{x_2^s, x_5^s\}$ and note that $x_5^s \in S$, or $x_1^{s+1}$ would be confounded with $x_3^{s+1}$.
            Now, if either $x_2^s \in S$, $x_4^s \in S$, or $x_1^s \in S$, we conclude that $|\varphi_s \cap S| \geq 3$; the first case is obvious, the second implies that at least one of $\{x_1^s, x_2^s, x_3^s\}$ must be in $S$ so $x_2^s$ is dominated, while the third that one of of $\{x_2^s, x_3^s, x_4^s\}$ or $x_3^s$ would not be dominated.
            As such, we conclude that $\varphi_s \cap S = \{x_3^s, x_5^s\}$, but now we may pick $r < s$ to be the largest index where $\varphi_r \cap S \neq \{x_3^r, x_5^r\}$ and, by a very similar argument to the case $\varphi_{\tp - 1} \cap S = \{x_3^{\tp-1}, x_5^{\tp-1}\}$, we conclude that $|\varphi_s \cap S| \geq 3$.
            
            As such, we add $x_2^i$ and $x_4^i$ to $S'$ for every $\varphi_i \in \varphi^*(P)$, and $x_1^\tp$ to $S'$, as in Figure~\ref{fig:normalize_ml4-2}.
            Since there is some $\varphi_\ell \in \varphi^*(P)$ with three vertices in $S$, $|S'| \leq |S|$.
            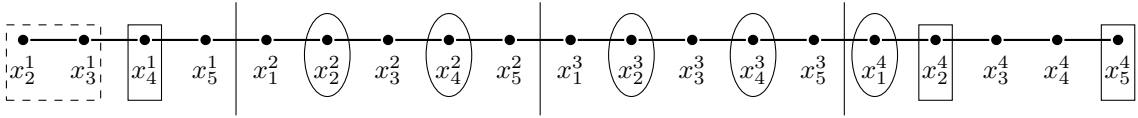
\begin{figure}[!htb]
                \hspace{-1.5cm}
                    \begin{tikzpicture}[scale=1]
                        %\draw[help lines] (-5,-5) grid (5,5);
                        \GraphInit[unit=3,vstyle=Normal]
                        \SetVertexNormal[Shape=circle, FillColor=black, MinSize=3pt]
                        \tikzset{VertexStyle/.append style = {inner sep = \inners, outer sep = \outers}}
                        \SetVertexLabelOut
                        %\draw (0, 1) -- (0, -1);
                        
                        \foreach \j in {1} {
                            \foreach \i in {2,3,4,5} {
                                \pgfmathsetmacro{\x}{(5*(\j-1) + \i)*0.8}
                               
                                \ifthenelse{\i = 1 \OR \i = 4} {
                                    \begin{scope}[xshift=\x cm]
                                        \draw (-0.22, -0.8) rectangle (0.22, 0.2);
                                    \end{scope}
                                }{}
                                \Vertex[x=\x, y=0, Lpos=270,Math, L={x_\i^\j}]{x_\i\j}
                            }
                        }
                        \foreach \j in {2,3} {
                            \foreach \i in {1,2,3,4,5} {
                                \pgfmathsetmacro{\x}{(5*(\j-1) + \i)*0.8}
                               
                                \ifthenelse{\i = 2 \OR \i = 4} {
                                    \begin{scope}[xshift=\x cm]
                                        \draw (0,-0.2) ellipse (0.3cm and 0.55cm);
                                    \end{scope}
                                }{}
                                \Vertex[x=\x, y=0, Lpos=270,Math, L={x_\i^\j}]{x_\i\j}
                            }
                        }
                        \foreach \j in {4} {
                            \foreach \i in {1,2,3,4,5} {
                                \pgfmathsetmacro{\x}{(5*(\j-1) + \i)*0.8}
                                \ifthenelse{\i = 1}{
                                    \begin{scope}[xshift=\x cm]
                                        \draw (0,-0.2) ellipse (0.3cm and 0.55cm);
                                    \end{scope}
                                }{}
                                \ifthenelse{\i = 2 \OR \i = 5} {
                                    \begin{scope}[xshift=\x cm]
                                        \draw (-0.22, -0.8) rectangle (0.22, 0.2);
                                    \end{scope}
                                }{}
                                \Vertex[x=\x, y=0, Lpos=270,Math, L={x_\i^\j}]{x_\i\j}
                            }
                        }
                        \draw (4.4, -1) -- (4.4,0.5);
                        \draw (8.4, -1) -- (8.4,0.5);
                        \draw (12.4, -1) -- (12.4,0.5);
                        
                        \draw[dashed] (1.38, -0.8) rectangle (2.62, 0.2);
                        %\draw[dashed] (12.58, -0.8) rectangle (13.82, 0.2);

                        \Edges(x_21, x_31, x_41, x_51, x_12, x_22, x_32, x_42, x_52, x_13, x_23, x_33, x_43, x_53, x_14, x_24, x_34, x_44, x_54)
                        %\Edge(x_11)(x_45)
                    \end{tikzpicture}
            \caption{Case~\ref{case:normalize_ml4-2} of the proof of Lemma~\ref{lem:normalize_ml}.
            Solid boxes indicate that the boxed vertex is in $S \cap S'$, while solid ellipses indicate vertices in $S'$ but not necessarily in $S$.
            The large dashed box indicates that at least one of the boxed vertices is in the solution.\label{fig:normalize_ml4-2}}
            \end{figure}
        \end{enumerate}
        
        \item{\label{case:normalize_ml-t}} If $x_5^1 \notin S$, $x_4^1 \in S$, and $x_3^\tp \in S$.
        To see that there is some $\varphi_\ell \in \varphi^*(P)$ that satisfies $|\varphi_\ell \cap S| \geq 3$, let $r > 1$ be the largest integer where $\varphi_r \cap S \neq \{x_3^r, x_5^r\}$, and note that $x_5^r \in S$, otherwise $x_5^{r-1}$ is not dominated.
        If $x_3^r \in S$ we are done, but we must have $x_4^r \in S$, otherwise $x_4^r$ is confounded with $x_1^{r-1}$, and that either $x_1^r$ or $x_2^r$ must also be in $S$, or $x_2^r$ would not be dominated by $S$.
        Consequently, if we add $x_1^\tp$ to $S'$, then $\rev(\varphi(P))$ satisfies Case~\ref{case:normalize_ml1}, and have that $|S'| \leq |S|$; we present an example in Figure~\ref{fig:normalize_ml-t}.
        \begin{figure}[!htb]
            \centering
                \begin{tikzpicture}[scale=1]
                    %\draw[help lines] (-5,-5) grid (5,5);
                    \GraphInit[unit=3,vstyle=Normal]
                    \SetVertexNormal[Shape=circle, FillColor=black, MinSize=3pt]
                    \tikzset{VertexStyle/.append style = {inner sep = \inners, outer sep = \outers}}
                    \SetVertexLabelOut
                    %\draw (0, 1) -- (0, -1);
                    
                    \foreach \j in {1} {
                        \foreach \i in {4,5} {
                            \pgfmathsetmacro{\x}{(5*(\j-1) + \i)*0.8}
                           
                            \ifthenelse{\i = 1 \OR \i = 4} {
                                \begin{scope}[xshift=\x cm]
                                    \draw (-0.22, -0.8) rectangle (0.22, 0.2);
                                \end{scope}
                            }{}
                            \Vertex[x=\x, y=0, Lpos=270,Math, L={x_\i^\j}]{x_\i\j}
                        }
                    }
                    \foreach \j in {2,3} {
                        \foreach \i in {1,2,3,4,5} {
                            \pgfmathsetmacro{\x}{(5*(\j-1) + \i)*0.8}
                           
                            \ifthenelse{\i = 1 \OR \i = 4} {
                                \begin{scope}[xshift=\x cm]
                                    \draw (0,-0.2) ellipse (0.3cm and 0.55cm);
                                \end{scope}
                            }{}
                            \Vertex[x=\x, y=0, Lpos=270,Math, L={x_\i^\j}]{x_\i\j}
                        }
                    }
                    \foreach \j in {4} {
                        \foreach \i in {1,2,3,4,5} {
                            \pgfmathsetmacro{\x}{(5*(\j-1) + \i)*0.8}
                            \ifthenelse{\i = 1}{
                                \begin{scope}[xshift=\x cm]
                                    \draw (0,-0.2) ellipse (0.3cm and 0.55cm);
                                \end{scope}
                            }{}
                            \ifthenelse{\i = 3 \OR \i = 5} {
                                \begin{scope}[xshift=\x cm]
                                    \draw (-0.22, -0.8) rectangle (0.22, 0.2);
                                \end{scope}
                            }{}
                            \Vertex[x=\x, y=0, Lpos=270,Math, L={x_\i^\j}]{x_\i\j}
                        }
                    }
                    \draw (4.4, -1) -- (4.4,0.5);
                    \draw (8.4, -1) -- (8.4,0.5);
                    \draw (12.4, -1) -- (12.4,0.5);
                    
                    %\draw[dashed] (1.38, -0.8) rectangle (2.62, 0.2);
                    %\draw[dashed] (12.58, -0.8) rectangle (13.82, 0.2);

                    \Edges(x_41, x_51, x_12, x_22, x_32, x_42, x_52, x_13, x_23, x_33, x_43, x_53, x_14, x_24, x_34, x_44, x_54)
                    %\Edge(x_11)(x_45)
                \end{tikzpicture}
        \caption{Case~\ref{case:normalize_ml-t} of the proof of Lemma~\ref{lem:normalize_ml}.
        Solid boxes indicate that the boxed vertex is in $S \cap S'$, while solid ellipses indicate vertices in $S'$ but not necessarily in $S$.\label{fig:normalize_ml-t}}
        \end{figure}
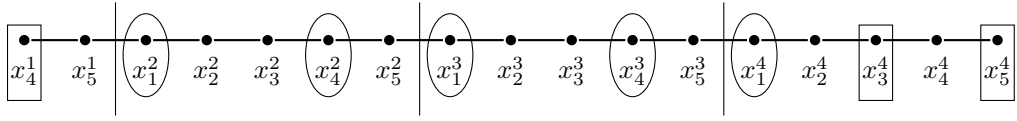
        
        \item{\label{case:normalize_ml5}} If $x_5^1 \notin S$, $N_S[x_3^1] = \{x_4^1\}$, and $x_1^\tp \in S$, then $\rev(\varphi)$ satisfies the condition of Case~\ref{case:normalize_ml1}, i.e. we complete $S'$  by adding $x_1^i$ and $x_4^i$ to it, as shown in Figure~\ref{fig:normalize_ml5}.
        
        \begin{figure}[!htb]
                \centering
                    \begin{tikzpicture}[scale=1]
                        %\draw[help lines] (-5,-5) grid (5,5);
                        \GraphInit[unit=3,vstyle=Normal]
                        \SetVertexNormal[Shape=circle, FillColor=black, MinSize=3pt]
                        \tikzset{VertexStyle/.append style = {inner sep = \inners, outer sep = \outers}}
                        \SetVertexLabelOut
                        %\draw (0, 1) -- (0, -1);
                        
                        \foreach \j in {1} {
                            \foreach \i in {1,2,3,4,5} {
                                \pgfmathsetmacro{\x}{(5*(\j-1) + \i)*0.8}
                               
                                \ifthenelse{\i = 1 \OR \i = 4} {
                                    \begin{scope}[xshift=\x cm]
                                        \draw (-0.22, -0.8) rectangle (0.22, 0.2);
                                    \end{scope}
                                }{}
                                \Vertex[x=\x, y=0, Lpos=270,Math, L={x_\i^\j}]{x_\i\j}
                            }
                        }
                        \foreach \j in {2,3} {
                            \foreach \i in {1,2,3,4,5} {
                                \pgfmathsetmacro{\x}{(5*(\j-1) + \i)*0.8}
                               
                                \ifthenelse{\i = 1 \OR \i = 4} {
                                    \begin{scope}[xshift=\x cm]
                                        \draw (0,-0.2) ellipse (0.3cm and 0.55cm);
                                    \end{scope}
                                }{}
                                \Vertex[x=\x, y=0, Lpos=270,Math, L={x_\i^\j}]{x_\i\j}
                            }
                        }
                        \foreach \j in {4} {
                            \foreach \i in {1} {
                                \pgfmathsetmacro{\x}{(5*(\j-1) + \i)*0.8}
                                \ifthenelse{\i = 0}{
                                    \begin{scope}[xshift=\x cm]
                                        \draw (0,-0.2) ellipse (0.3cm and 0.55cm);
                                    \end{scope}
                                }{}
                                \ifthenelse{\i = 1} {
                                    \begin{scope}[xshift=\x cm]
                                        \draw (-0.22, -0.8) rectangle (0.22, 0.2);
                                    \end{scope}
                                }{}
                                \Vertex[x=\x, y=0, Lpos=270,Math, L={x_\i^\j}]{x_\i\j}
                            }
                        }
                        \draw (4.4, -1) -- (4.4,0.5);
                        \draw (8.4, -1) -- (8.4,0.5);
                        \draw (12.4, -1) -- (12.4,0.5);
                        
                        %\draw[dashed] (1.38, -0.8) rectangle (2.62, 0.2);
                        %\draw[dashed] (12.58, -0.8) rectangle (13.82, 0.2);

                        \Edges(x_11, x_21, x_31, x_41, x_51, x_12, x_22, x_32, x_42, x_52, x_13, x_23, x_33, x_43, x_53, x_14)
                        %\Edge(x_11)(x_45)
                    \end{tikzpicture}
            \caption{Case~\ref{case:normalize_ml5} of the proof of Lemma~\ref{lem:normalize_ml}.
            Solid boxes indicate that the boxed vertex is in $S \cap S'$, while solid ellipses indicate vertices in $S'$ but not necessarily in $S$.
            The large dashed box indicates that at least one of the boxed vertices is in the solution.\label{fig:normalize_ml5}}
            \end{figure}
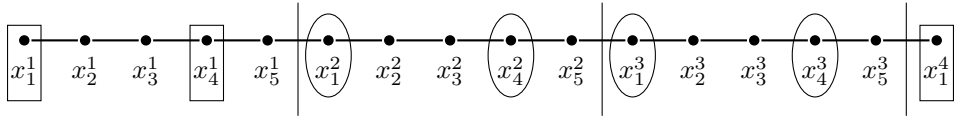
            
        \item{\label{case:normalize_ml6}} If $x_5^1 \notin S$, $N_S[x_3^1] = \{x_4^1\}$, and $x_2^\tp \in S$, then we branch in two subcases.
        
        \begin{enumerate}
            \item{\label{case:normalize_ml6-1}} If $N_S[x_3^\tp] \setminus \{x_2^\tp\} \neq \emptyset$, then $\rev(\varphi(P))$ satisfies Case~\ref{case:normalize_ml4-2}, so we add $x_5^1$, $x_2^i$, and $x_4^i$, for every $\varphi^*(P)$, to $S'$, as exemplified in Figure~\ref{fig:normalize_ml6-1}.
            \begin{figure}[!htb]
                \hspace{-1.5cm}
                    \begin{tikzpicture}[scale=1]
                        %\draw[help lines] (-5,-5) grid (5,5);
                        \GraphInit[unit=3,vstyle=Normal]
                        \SetVertexNormal[Shape=circle, FillColor=black, MinSize=3pt]
                        \tikzset{VertexStyle/.append style = {inner sep = \inners, outer sep = \outers}}
                        \SetVertexLabelOut
                        %\draw (0, 1) -- (0, -1);
                        
                        \foreach \j in {1} {
                            \foreach \i in {1,2,3,4,5} {
                                \pgfmathsetmacro{\x}{(5*(\j-1) + \i)*0.8}
                               
                                \ifthenelse{\i = 5} {
                                    \begin{scope}[xshift=\x cm]
                                        \draw (0,-0.2) ellipse (0.3cm and 0.55cm);
                                    \end{scope}
                                }{}
                                \ifthenelse{\i = 1 \OR \i = 4} {
                                    \begin{scope}[xshift=\x cm]
                                        \draw (-0.22, -0.8) rectangle (0.22, 0.2);
                                    \end{scope}
                                }{}
                                \Vertex[x=\x, y=0, Lpos=270,Math, L={x_\i^\j}]{x_\i\j}
                            }
                        }
                        \foreach \j in {2,3} {
                            \foreach \i in {1,2,3,4,5} {
                                \pgfmathsetmacro{\x}{(5*(\j-1) + \i)*0.8}
                               
                                \ifthenelse{\i = 2 \OR \i = 4} {
                                    \begin{scope}[xshift=\x cm]
                                        \draw (0,-0.2) ellipse (0.3cm and 0.55cm);
                                    \end{scope}
                                }{}
                                \Vertex[x=\x, y=0, Lpos=270,Math, L={x_\i^\j}]{x_\i\j}
                            }
                        }
                        \foreach \j in {4} {
                            \foreach \i in {1,2,3,4} {
                                \pgfmathsetmacro{\x}{(5*(\j-1) + \i)*0.8}
                                \ifthenelse{\i = 0}{
                                    \begin{scope}[xshift=\x cm]
                                        \draw (0,-0.2) ellipse (0.3cm and 0.55cm);
                                    \end{scope}
                                }{}
                                \ifthenelse{\i = 2 \OR \i = 5} {
                                    \begin{scope}[xshift=\x cm]
                                        \draw (-0.22, -0.8) rectangle (0.22, 0.2);
                                    \end{scope}
                                }{}
                                \Vertex[x=\x, y=0, Lpos=270,Math, L={x_\i^\j}]{x_\i\j}
                            }
                        }
                        \draw (4.4, -1) -- (4.4,0.5);
                        \draw (8.4, -1) -- (8.4,0.5);
                        \draw (12.4, -1) -- (12.4,0.5);
                        
                        %\draw[dashed] (1.38, -0.8) rectangle (2.62, 0.2);
                        \draw[dashed] (14.18, -0.8) rectangle (15.42, 0.2);

                        \Edges(x_11, x_21, x_31, x_41, x_51, x_12, x_22, x_32, x_42, x_52, x_13, x_23, x_33, x_43, x_53, x_14, x_24, x_34, x_44)
                        %\Edge(x_11)(x_45)
                    \end{tikzpicture}
            \caption{Case~\ref{case:normalize_ml6-1} of the proof of Lemma~\ref{lem:normalize_ml}.
            Solid boxes indicate that the boxed vertex is in $S \cap S'$, while solid ellipses indicate vertices in $S'$ but not necessarily in $S$.
            The large dashed box indicates that at least one of the boxed vertices is in the solution.\label{fig:normalize_ml6-1}}
            \end{figure}
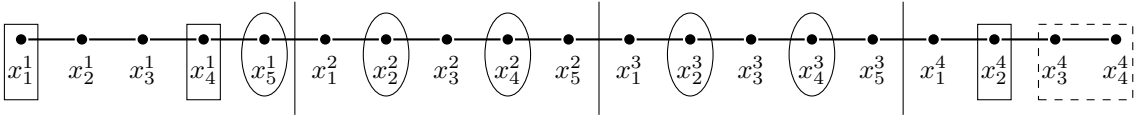
            
            \item{\label{case:normalize_ml6-2}} On the other hand, if $N_{S}[x_3^\tp] = \{x_2^\tp\}$, we again claim that there is some $\varphi_\ell \in \varphi^*(P)$ with at least three vertices in $S$.
            
            To see that this is the case, let $r$ be the smallest integer where $\varphi_r \in \varphi^*(P)$ satisfies $\varphi_r \cap S \neq \{x_1^r, x_4^r\}$, and note that $x_1^r \in S$, otherwise $x_3^{r-1}$ is confounded with $x_5^{r-1}$.
            If $r = \tp - 1$, we have that $x_5^{\tp - 1} \in S$, or $x_1^\tp$ is confounded with $x_3^\tp$.
            Neither $x_1^{\tp - 1}$ nor $x_5^{\tp - 1}$, however, dominate $x_3^3$, so at at least one of $\{x_2^{\tp - 1}, x_3^{\tp - 1}, x_4^{\tp - 1}\}$ must be in $S$, which implies $\varphi_{\tp - 1} \cap S \supset \{x_1^{\tp - 1}, x_5^{\tp - 1}\}$.
            If $r < \tp - 1$ we may assume that $x_4^r \notin S$, or we would be done.
            So either $x_2^r \in S$, which implies that either $x_3^r$ or $x_4^r$ are in $S$ and that $|\varphi_r \cap S| \geq 3$, or $x_3^r \in S$, which we now assume to be the case.
            Let $s$ be the smallest index at least as large as $r$ that has $\varphi_s \cap S \neq \{x_1^s, x_3^s\}$; note that $s < \tp$, otherwise $x_1^\tp$ and $x_3^\tp$ are confounded, and that $x_1^s \in S$, otherwise $x_5^{s - 1}$ is not dominated.
            If $x_3^s \in S$, we are done, since this would imply $|\varphi_s \cap S| \geq 3$, but in this case we need $x_2^s \in S$, otherwise $x_2^s$ and $x_5^{s-1}$ are confounded.
            However, we also require one of $\{x_4^s, x_5^s\}$ to be in $S$, otherwise $x_4^s$ is not dominated by any vertex, so $|\varphi_s \cap S| \geq 3$.
            
            As such, we add $x_1^i$ and $x_4^i$ to $S'$, for every $\varphi_i \in \varphi^*(P)$, and also add $x_1^\tp$ to $S$.
            This maintains $|S'| = |S|$ and that $S'$ is a solution follows from our previous argumentation and that $S$ is a solution; see Figure~\ref{fig:normalize_ml5} for an example.
            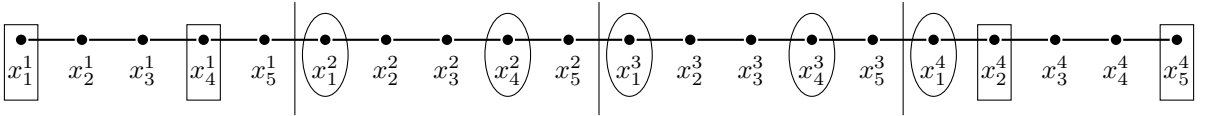
\begin{figure}[!htb]
                \hspace{-1.5cm}
                        \begin{tikzpicture}[scale=1]
                            %\draw[help lines] (-5,-5) grid (5,5);
                            \GraphInit[unit=3,vstyle=Normal]
                            \SetVertexNormal[Shape=circle, FillColor=black, MinSize=3pt]
                            \tikzset{VertexStyle/.append style = {inner sep = \inners, outer sep = \outers}}
                            \SetVertexLabelOut
                            %\draw (0, 1) -- (0, -1);
                            
                            \foreach \j in {1} {
                                \foreach \i in {1,2,3,4,5} {
                                    \pgfmathsetmacro{\x}{(5*(\j-1) + \i)*0.8}
                                   
                                    \ifthenelse{\i = 1 \OR \i = 4} {
                                        \begin{scope}[xshift=\x cm]
                                            \draw (-0.22, -0.8) rectangle (0.22, 0.2);
                                        \end{scope}
                                    }{}
                                    \Vertex[x=\x, y=0, Lpos=270,Math, L={x_\i^\j}]{x_\i\j}
                                }
                            }
                            \foreach \j in {2,3} {
                                \foreach \i in {1,2,3,4,5} {
                                    \pgfmathsetmacro{\x}{(5*(\j-1) + \i)*0.8}
                                   
                                    \ifthenelse{\i = 1 \OR \i = 4} {
                                        \begin{scope}[xshift=\x cm]
                                            \draw (0,-0.2) ellipse (0.3cm and 0.55cm);
                                        \end{scope}
                                    }{}
                                    \Vertex[x=\x, y=0, Lpos=270,Math, L={x_\i^\j}]{x_\i\j}
                                }
                            }
                            \foreach \j in {4} {
                                \foreach \i in {1,2,3,4,5} {
                                    \pgfmathsetmacro{\x}{(5*(\j-1) + \i)*0.8}
                                    \ifthenelse{\i = 1}{
                                        \begin{scope}[xshift=\x cm]
                                            \draw (0,-0.2) ellipse (0.3cm and 0.55cm);
                                        \end{scope}
                                    }{}
                                    \ifthenelse{\i = 2 \OR \i = 5} {
                                        \begin{scope}[xshift=\x cm]
                                            \draw (-0.22, -0.8) rectangle (0.22, 0.2);
                                        \end{scope}
                                    }{}
                                    \Vertex[x=\x, y=0, Lpos=270,Math, L={x_\i^\j}]{x_\i\j}
                                }
                            }
                            \draw (4.4, -1) -- (4.4,0.5);
                            \draw (8.4, -1) -- (8.4,0.5);
                            \draw (12.4, -1) -- (12.4,0.5);
                            
                            %\draw[dashed] (1.38, -0.8) rectangle (2.62, 0.2);
                            %\draw[dashed] (12.58, -0.8) rectangle (13.82, 0.2);

                            \Edges(x_11, x_21, x_31, x_41, x_51, x_12, x_22, x_32, x_42, x_52, x_13, x_23, x_33, x_43, x_53, x_14, x_24, x_34, x_44, x_54)
                            %\Edge(x_11)(x_45)
                        \end{tikzpicture}
                \caption{Case~\ref{case:normalize_ml6-2} of the proof of Lemma~\ref{lem:normalize_ml}.
                Solid boxes indicate that the boxed vertex is in $S \cap S'$, while solid ellipses indicate vertices in $S'$ but not necessarily in $S$.
                The large dashed box indicates that at least one of the boxed vertices is in the solution.\label{fig:normalize_ml6-2}}
            \end{figure}
        \end{enumerate}
        
        \item{\label{case:normalize_ml7}} If neither $x_4^1$ nor $x_5^1$ are in $S$, but $x_3^1 \in S$, and $x_1^\tp \in S$, note that $\rev(\varphi(P))$ satisfies Case~\ref{case:normalize_ml2}, so, for every $\varphi_i \in \varphi^*(P)$, add $x_1^i$ and $x_3^i$ to $S'$.
        If $\{x_3^\tp, x_4^\tp \} \cap S = \emptyset$, i.e. $\rev(\varphi(P))$ satisfies Case~\ref{case:normalize_ml2-2}, add $x_3^\tp$ to $S'$.

        \item{\label{case:normalize_ml8}} If neither $x_4^1$ nor $x_5^1$ are in $S$, but $x_3^1 \in S$, and $x_2^\tp \in S$, then $\rev(\varphi(P))$ satisfies Case~\ref{case:normalize_ml-t}, so by adding $x_5^1$, $x_2^i$, $x_5^i$ to $S'$, for every $\varphi_i \in \varphi^*(P)$, we guarantee that $S'$ is a solution and $|S'| \leq |S|$.
        An example is given in Figure~\ref{fig:normalize_ml8}.
        \begin{figure}[!htb]
            \centering
                \begin{tikzpicture}[scale=1]
                    %\draw[help lines] (-5,-5) grid (5,5);
                    \GraphInit[unit=3,vstyle=Normal]
                    \SetVertexNormal[Shape=circle, FillColor=black, MinSize=3pt]
                    \tikzset{VertexStyle/.append style = {inner sep = \inners, outer sep = \outers}}
                    \SetVertexLabelOut
                    %\draw (0, 1) -- (0, -1);
                    
                    \foreach \j in {1} {
                        \foreach \i in {1,2,3,4,5} {
                            \pgfmathsetmacro{\x}{(5*(\j-1) + \i)*0.8}
                           
                            \ifthenelse{\i = 5} {
                                \begin{scope}[xshift=\x cm]
                                    \draw (0,-0.2) ellipse (0.3cm and 0.55cm);
                                \end{scope}
                            }{}
                            \ifthenelse{\i = 1 \OR \i = 3} {
                                \begin{scope}[xshift=\x cm]
                                    \draw (-0.22, -0.8) rectangle (0.22, 0.2);
                                \end{scope}
                            }{}
                            \Vertex[x=\x, y=0, Lpos=270,Math, L={x_\i^\j}]{x_\i\j}
                        }
                    }
                    \foreach \j in {2,3} {
                        \foreach \i in {1,2,3,4,5} {
                            \pgfmathsetmacro{\x}{(5*(\j-1) + \i)*0.8}
                           
                            \ifthenelse{\i = 2 \OR \i = 5} {
                                \begin{scope}[xshift=\x cm]
                                    \draw (0,-0.2) ellipse (0.3cm and 0.55cm);
                                \end{scope}
                            }{}
                            \Vertex[x=\x, y=0, Lpos=270,Math, L={x_\i^\j}]{x_\i\j}
                        }
                    }
                    \foreach \j in {4} {
                        \foreach \i in {1,2} {
                            \pgfmathsetmacro{\x}{(5*(\j-1) + \i)*0.8}
                            \ifthenelse{\i = 2} {
                                \begin{scope}[xshift=\x cm]
                                    \draw (-0.22, -0.8) rectangle (0.22, 0.2);
                                \end{scope}
                            }{}
                            \Vertex[x=\x, y=0, Lpos=270,Math, L={x_\i^\j}]{x_\i\j}
                        }
                    }
                    \draw (4.4, -1) -- (4.4,0.5);
                    \draw (8.4, -1) -- (8.4,0.5);
                    \draw (12.4, -1) -- (12.4,0.5);
                    
                    %\draw[dashed] (1.38, -0.8) rectangle (2.62, 0.2);
                    %\draw[dashed] (14.18, -0.8) rectangle (15.42, 0.2);

                    \Edges(x_11, x_12, x_31, x_41, x_51, x_12, x_22, x_32, x_42, x_52, x_13, x_23, x_33, x_43, x_53, x_14, x_24)
                    %\Edge(x_11)(x_45)
                \end{tikzpicture}
        \caption{Case~\ref{case:normalize_ml8} of the proof of Lemma~\ref{lem:normalize_ml}.
        Solid boxes indicate that the boxed vertex is in $S \cap S'$, while solid ellipses indicate vertices in $S'$ but not necessarily in $S$.
        The large dashed box indicates that at least one of the boxed vertices is in the solution.\label{fig:normalize_ml8}}
        \end{figure}
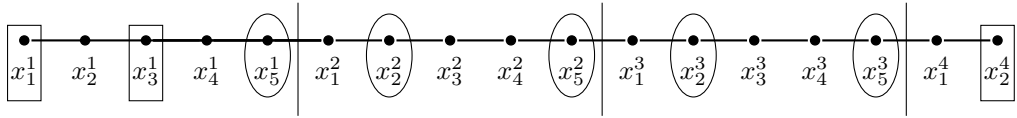
        
        \item{\label{case:normalize_ml9}} If neither $x_4^1$ nor $x_5^1$ are in $S$, but $x_3^1 \in S$, and $x_3^\tp \in S$.
        Towards showing that there is some $\varphi_\ell \in \varphi^*$ with three vertices in $S$, again let $r > 1$ be the smallest index where $\varphi_r \cap S \neq \{x_1^r, x_3^r\}$ and note, once more, that $x_1^r \in S$, otherwise $x_5^{r-1}$ is not dominated.
        Since $x_3^r \notin S$, otherwise we would be done, $x_2^r \in S$, but one more vertex of $\varphi_r$ is required to dominate $x_4^r$.
        As such, by adding $x_5^1$, $x_3^i$, and $x_5^i$, for every $\varphi_i \in \varphi^*(P)$, by Case~\ref{case:normalize_ml2-1}, we have a solution to $(G,d)$, as shown in Figure~\ref{fig:normalize_ml9}.
        \begin{figure}[!htb]
                \hspace{-1.5cm}
                \begin{tikzpicture}[scale=1]
                    %\draw[help lines] (-5,-5) grid (5,5);
                    \GraphInit[unit=3,vstyle=Normal]
                    \SetVertexNormal[Shape=circle, FillColor=black, MinSize=3pt]
                    \tikzset{VertexStyle/.append style = {inner sep = \inners, outer sep = \outers}}
                    \SetVertexLabelOut
                    %\draw (0, 1) -- (0, -1);
                    
                    \foreach \j in {1} {
                        \foreach \i in {1,2,3,4,5} {
                            \pgfmathsetmacro{\x}{(5*(\j-1) + \i)*0.8}
                           
                            \ifthenelse{\i = 5} {
                                \begin{scope}[xshift=\x cm]
                                    \draw (0,-0.2) ellipse (0.3cm and 0.55cm);
                                \end{scope}
                            }{}
                            \ifthenelse{\i = 1 \OR \i = 3} {
                                \begin{scope}[xshift=\x cm]
                                    \draw (-0.22, -0.8) rectangle (0.22, 0.2);
                                \end{scope}
                            }{}
                            \Vertex[x=\x, y=0, Lpos=270,Math, L={x_\i^\j}]{x_\i\j}
                        }
                    }
                    \foreach \j in {2,3} {
                        \foreach \i in {1,2,3,4,5} {
                            \pgfmathsetmacro{\x}{(5*(\j-1) + \i)*0.8}
                           
                            \ifthenelse{\i = 3 \OR \i = 5} {
                                \begin{scope}[xshift=\x cm]
                                    \draw (0,-0.2) ellipse (0.3cm and 0.55cm);
                                \end{scope}
                            }{}
                            \Vertex[x=\x, y=0, Lpos=270,Math, L={x_\i^\j}]{x_\i\j}
                        }
                    }
                    \foreach \j in {4} {
                        \foreach \i in {1,2,3,4,5} {
                            \pgfmathsetmacro{\x}{(5*(\j-1) + \i)*0.8}
                            \ifthenelse{\i = 3 \OR \i = 5} {
                                \begin{scope}[xshift=\x cm]
                                    \draw (-0.22, -0.8) rectangle (0.22, 0.2);
                                \end{scope}
                            }{}
                            \Vertex[x=\x, y=0, Lpos=270,Math, L={x_\i^\j}]{x_\i\j}
                        }
                    }
                    \draw (4.4, -1) -- (4.4,0.5);
                    \draw (8.4, -1) -- (8.4,0.5);
                    \draw (12.4, -1) -- (12.4,0.5);
                    
                    %\draw[dashed] (1.38, -0.8) rectangle (2.62, 0.2);
                    %\draw[dashed] (14.18, -0.8) rectangle (15.42, 0.2);

                    \Edges(x_11, x_12, x_31, x_41, x_51, x_12, x_22, x_32, x_42, x_52, x_13, x_23, x_33, x_43, x_53, x_14, x_24, x_34, x_44, x_54)
                    %\Edge(x_11)(x_45)
                \end{tikzpicture}
        \caption{Case~\ref{case:normalize_ml9} of the proof of Lemma~\ref{lem:normalize_ml}.
        Solid boxes indicate that the boxed vertex is in $S \cap S'$, while solid ellipses indicate vertices in $S'$ but not necessarily in $S$.\label{fig:normalize_ml9}}
        \end{figure}
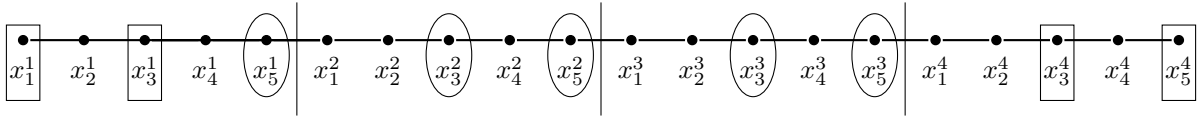
    \end{enumerate}
    The proof is now complete since at least one of $\{x_3^1, x_4^1, x_5^1\}$ and at least one of $\{x_1^\tp, x_2^\tp, x_3^\tp\}$ must be in $S$ and all possible cases have been analyzed.
\end{proof}

With Lemma~\ref{lem:normalize_ml} in hand, we are now ready to state the unique reduction rule of our kernelization algorithm.

\begin{reduction}
    \label{rule:maxleaf}
    Let $(G,d)$ be an instance of \pname{Locating-Dominating Set}.
    If $P \in \mathcal{P}_2(G)$ has at least 20 vertices and $\varphi(P)$ is a 5-sectioning of $P$, remove all vertices in $\varphi^*(P)$ from $G$, add five vertices $X^P = \{x_i^P\}_{i \in [5]}$ to the resulting graph and the edges $x_i^Px_{i+1}^P$, for every $i \in [4]$, $x_5^Px_1^\tp$, and $x_s^1x_1^P$, where $s = |\varphi_1|$; finally, set $d'= d - 2(|\varphi^*(P)| - 1)$.
\end{reduction}

\begin{proof}(of safeness of Rule~\ref{rule:maxleaf})
    Throughout this proof, let $G'$ be the graph obtained by applying Rule~\ref{rule:maxleaf}.
    Let $S$ be a solution to $(G,d)$ to which we have already applied Lemma~\ref{lem:normalize_ml}, i.e. if $\varphi_i \in \varphi^*(P)$ has $\varphi_i \cap S = \{x_a^i, x_b^i\}$, then $\varphi_j \cap S = \{x_a^j, x_b^j\}$ for every $\varphi_j \in \varphi^*(P)$.
    We claim that $S' = (S \cap V(G')) \cup \{x_a^P, x_b^P\}$ is a solution to $(G', d')$.
    To see that this is the case, note that, for every vertex $v \in V(G') \setminus (\{x_s^1, x_1^\tp\} \cup X^P)$, $N_S[v] = N_{S'}[v]$ and, if $x_s^1$ (resp. $x_1^\tp$) has a neighbor in $\varphi_2$ (resp. $\varphi_{\tp-1}$, then it has a neighbor $X^P$ (resp. $X^P$);
    as such, no vertex in $V(G') \setminus X^P$ can be confounded with a vertex of $G'$ under $S'$.
    Moreover, since no vertex in $\varphi^*(P)$ was confounded with another vertex of $G$ under $S$, it holds that the neighborhood of each $x_j^P \in X^P$ in $S'$ is also unique.
    Finally, since $|S| \leq d$, we have $|S'| = |S| - 2(|\varphi^*(P)| - 1) \leq d - 2(|\varphi^*(P)| - 1) \leq d'$.
    
    For the converse, let $S'$ be a solution to $(G', d')$ normalized by Lemma~\ref{lem:normalize_ml} and $S = (S' \setminus X^P) \cup \bigcup_{1 \leq i \leq \tp - 1}\{x_a^i, x_b^i\}$, where $a$ and $b$ satisfy $S' \cap X^P = \{x_a^P, x_b^P\}$.
    We claim that $S$ is a solution to $(G,d)$ because: (i) $N_S[v] = N_{S'}[v]$ for every $v \in V(G) \cap V(G') \setminus \{x_s^1, x_1^\tp\}$; (ii) if $x_s^1$ (resp. $x_1^\tp$) is confounded with $x_{s-2}^1$ (resp. $x_3^\tp$) or $x_2^2$ (resp. $x_4^{\tp-1}$) under $S$, then $x_s^1$ (resp. $x_1^\tp$) would be confounded with $x_{s-2}^1$ (resp. $x_3^\tp$) or $x_2^P$ (resp. $x_4^P$) under $S$; (iii) $x_a^i \varphi^i$ cannot be confounded with $x_b^i \in \varphi_i$ under $S$, since the normalization given by Lemma~\ref{lem:normalize_ml} imposes that no two consecutive vertices of $X^P$ are in $S'$, so no two consecutive vertices of $\varphi^*(P)$ are in $S$; (iv) $x_4^i$ (resp. $x_5^i$) cannot be confounded with $x_1^{i+1}$ (resp. $x_2^{i+1}$) since, by the case analysis of Lemma~\ref{lem:normalize_ml}, either at least one of them is in $S$, or at least one of them has a non-common neighbor in the solution.
    As to the size of $S$, since $|S'| \leq d' = d - 2(|\varphi^*(P) - 1)$ and $|S| = |S'| + 2(|\varphi^*(P) - 1)$, we have $|S| \leq d$, as desired.
\end{proof}

\begin{theorem}
    \label{thm:max_leaf}
    When parameterized by the maximum leaf number, \pname{Locating-Dominating Set} admits a kernel with at most $108k + \floor{\frac{k}{2}}$ vertices that can be computed in $\bigO{k(n + m)}$, where $n = |V(G)|$ and $m = |E(G)|$.
\end{theorem}

\begin{proof}
    For each path $P \in \mathcal{P}_2(G)$ with at least 20 vertices, apply Rule~\ref{rule:maxleaf}, can easily be done in $\bigO{n+m}$ time.
    By Lemma~\ref{lem:ml_bound}, the rule is applied at most $\bigO{k}$ times, yielding the claimed complexity.
    As to the size of the kernel, again by Lemma~\ref{lem:ml_bound}, there are at most $4k + \floor{\frac{k}{2}} - 1$ paths in $\mathcal{P}_2(G)$ and, since Rule~\ref{rule:maxleaf} is not applicable, each path in $\mathcal{P}_2(G)$ has at most 19 vertices, so we conclude that the reduced instance has at most $4k + 19(5k + \floor{\frac{k}{2}} - 1) = 108k + \floor{\frac{k}{2}} - 1$ vertices.
\end{proof}
\section{Concluding Remarks}
In this work, we settled some algorithmic lower bounds on the natural parameter, showing that the naive $n^\bigO{d}$ time algorithm is optimal under the Exponential Time Hypothesis.
We also conducted a preliminary investigation into kernelization aspects of the \pname{Locating-Dominating Set} problem.
Our main results in this direction are the lower bounds presented in Section~\ref{sec:nokernel}, which together imply that most structural parameterizations of the graph parameter hierarchy~\cite{gpp} do not admit polynomial kernels under the complexity theoretic assumption that $\NP \nsubseteq \coNP/\poly$.
We observe that the bandwidth parameterization offers no hope of a polynomial kernel, since a trivial AND-cross-composition from \pname{Locating-Dominating Set} to itself that simply takes the disjoint union of the input instances suffices.
Of the few remaining structural parameterizations, we show that some polynomial cases do exist by exhibiting a linear kernel under the max leaf number parameterization.
As a primary open problem, we list the situation for the feedback edge set parameter: neither our compositions nor our kernel appear to be easily extended to this parameterization.
We are also interested in lower bounds for the running times of the \FPT\ algorithms, especially since the cliquewidth result described in~\cite{auger2010minimal} relies on Courcelle's Theorem, which yields algorithms whose dependency on the cliquewidth is an exponential tower.
Another viable research direction is the investigation of the \pname{Identifying Code} problem, which is usually studied alongside \pname{Locating-Dominating Set}; in this problem, we also want that vertices \textit{inside} the dominating set have a unique (closed) neighborhood in the solution.
As far as we could check, none of our approaches is applicable to \pname{Identifying Code}, which may yield interesting complexity differences between these sibling problems.
%----------------------------------------------------------------------------------
\bibliographystyle{abbrv}
\bibliography{main}
\end{document}